\documentclass[12pt]{article}
\usepackage{amsmath, bm}
\usepackage{amssymb, amsthm}
\usepackage{authblk, algorithm, algorithmic}
\usepackage{graphicx}
\usepackage{hyperref}
\usepackage{setspace}
\usepackage{natbib}
\usepackage{url} 
\usepackage{enumitem}
\usepackage{xcolor}
\usepackage{authblk}
\usepackage{bbm}
\usepackage{subcaption}

\usepackage{rotating}
\usepackage{geometry}
\usepackage{pdflscape}
\newcommand{\blind}{0}
\newtheorem{theorem}{Theorem}[section]
\newtheorem{proposition}[theorem]{Proposition}
\newtheorem{corollary}[theorem]{Corollary}

\newcommand{\argmin}{\operatornamewithlimits{argmin}}

\newcommand{\cG}{{\cal G}}

\begin{document}

\def\spacingset#1{\renewcommand{\baselinestretch}%
{#1}\small\normalsize} \spacingset{1}

\date{\today}
\if0\blind
{ 
  \title{\vspace{-4ex} \bf Generative Quantile Regression with \\ Variability Penalty}
\author[1]{Shijie Wang}
\author[2]{Minsuk Shin}
\author[1]{Ray Bai}

\affil[1]{\normalsize Department of Statistics, University of South Carolina, Columbia, SC 29208 \newline E-mails: \href{mailto:shijiew@email.sc.edu}{shijiew@email.sc.edu}, \href{mailto:rbai@mailbox.sc.edu}{rbai@mailbox.sc.edu}}
\affil[2]{\normalsize Gauss Labs, Palo Alto, CA 94301}

  \maketitle
} \fi

\if1\blind
{
  \begin{center}
    {\LARGE\bf Title}
\end{center}
} \fi

\begin{abstract}
Quantile regression and conditional density estimation can reveal structure that is missed by mean regression, such as multimodality and skewness. In this paper, we introduce a deep learning generative model for joint quantile estimation called Penalized Generative Quantile Regression (PGQR). Our approach simultaneously generates samples from many random quantile levels, allowing us to infer the conditional distribution of a response variable given a set of covariates. Our method employs a novel variability penalty to avoid the problem of vanishing variability, or memorization, in deep generative models. Further, we introduce a new family of partial monotonic neural networks (PMNN) to circumvent the problem of crossing quantile curves. A major benefit of PGQR is that it can be fit using a single optimization, thus bypassing the need to repeatedly train the model at multiple quantile levels or use computationally expensive cross-validation to tune the penalty parameter. We illustrate the efficacy of PGQR through extensive simulation studies and analysis of real datasets. Code to implement our method is available at \href{https://github.com/shijiew97/PGQR}{\tt https://github.com/shijiew97/PGQR}.
\end{abstract}

\noindent%

\spacingset{1.5} 
\section{Introduction}\label{sec:intro}
Quantile regression is a popular alternative to classical mean regression \citep{koenker1982robust}. For a response variable $Y \in \mathbb{R}$ and covariates $\bm{X} = (X_1, X_2, \ldots, X_p)^\top \in \mathbb{R}^{p}$, we define the $\tau$-th conditional quantile, $\tau \in (0,1)$, of $Y$ given $\bm{X}$ as
\begin{equation} \label{quantilefn}
	Q_{Y \mid \bm{X}} (\tau) = \inf \{ y: F_{Y \mid \bm{X}} (y) \geq \tau \},
\end{equation} 
where $F_{Y \mid \bm{X}}$ is the conditional cumulative distribution function of $Y$ given $\bm X$.
For a fixed quantile level $\tau \in (0, 1)$, linear quantile regression aims to model \eqref{quantilefn} as 
$Q_{Y \mid \bm{X}} (\tau) = \bm{X}^\top \bm{\beta}_{\tau}$. Linear quantile regression is less sensitive to outliers than least squares regression and is more robust when the assumption of independent and identically distributed 
 (iid) Gaussian residual
errors is violated \citep{koenker2001quantile,koenker2017handbook}. Thus, it is widely used for modeling  heterogeneous covariate-response associations.
To reveal complex nonlinear structures, \textit{nonparametric} quantile regression has also been proposed \citep{chaudhuri2002nonparametric, li2021nonparametric, koenker1994Biometrika}. For some function class $\mathcal{F}$, nonparametric quantile regression aims to estimate $f(\bm{X}, \tau) \in \mathcal{F}$ for $Q_{Y \mid \bm{X}}(\tau) = f(\bm{X}, \tau)$ at a given quantile level $\tau$.  
While quantile regression has traditionally focused on estimating $f(\bm{X}, \tau)$ a single $\tau$-th quantile, \textit{joint} quantile regression aims to estimate multiple quantile levels \textit{simultaneously} for a set of $K \geq 2$ quantiles $\{ \tau_1, \tau_2, \ldots, \tau_K \}$ 
\citep{zou2008composite, jiang2012oracle, xu2017composite}. 

Despite its flexibility, nonparametric quantile regression carries the risk that the estimated curves at different quantile levels might \textit{cross} each other. For instance, the estimate of the 95th conditional quantile of $Y_i$ given covariates $\bm{X}_i$ might be smaller than the estimate of the 90th conditional quantile. This is known as the \textit{crossing quantile} phenomenon. The left two panels of Figure \ref{fig:crossdata} illustrate the crossing problem. When quantile curves cross each other, the quantile estimates violate the laws of probability and are not reasonable. To tackle this issue, many approaches have been proposed, such as constraining the model space \citep{takeuchi2006nonparametric,sangnier2016joint,moon2021learning} or constraining the model parameters \citep{meinshausen2006quantile,cannon2018non}.

Recently in the area of deep learning, neural networks have also been applied to nonparametric quantile regression. \cite{zhong2023neural} introduced a semiparametric quantile regression method where the covariates of primary interest are modeled linearly, while all other covariates are modeled nonparametrically by a neural network.
\cite{shen2021deep} proposed a deep quantile regression (DQR) estimator to approximate the target conditional quantile function with compositional structure. Under the assumption that the conditional quantile function is a composition of low-dimentional functions, DQR is shown to achieve the minimax optimal convergence rate.

Researchers have also used neural networks to learn \emph{multiple} quantiles using composite quantile loss functions. \cite{cannon2018non} proposed the monotone composite quantile regression neural network (MCQRNN) and achieved non-crossing quantiles by imposing monotone increasing bias parameters for the neural network.  \cite{moon2021learning} recently introduced an efficient and scalable $\ell_1$-penalization algorithm that also estimates non-crossing quantiles with neural networks. However, existing neural network-based approaches for joint quantile estimation are restricted at a prespecified quantile set. If only several quantiles (e.g. the interquartiles $\tau \in \{ 0.25, 0.5, 0.75 \}$) are of interest, then these methods are adequate to produce desirable inference. Otherwise, one typically has to refit the model at new quantiles \textit{or} specify a large enough quantile candidate set in order to infer the full conditional density $p(Y \mid \bm{X})$. \cite{dabney2018implicit} introduced the implicit quantile network (IQN), which takes a grid of quantile levels as inputs and approximates the conditional density at \textit{any} new quantile level. 
Unfortunately, IQN \textit{cannot} guarantee that quantile functions do not cross each other.  
More recently, \cite{shen2022estimation} constructed a quantile regression process similar to IQN but which avoids crossing quantiles by imposing a positivity constraint on the first derivatives.

Joint quantile regression is closely related to the problem of \textit{conditional density estimation} (CDE) of $p(Y \mid \bm{X})$. Apart from traditional CDE methods that estimate the unknown probability density curve \citep{izbicki2017converting,https://doi.org/10.48550/arxiv.1906.07177}, several deep \textit{generative} approaches besides IQN have also been proposed. \cite{zhou2022deep} introduced the generative conditional distribution sampler (GCDS), while \cite{liu2021wasserstein} introduced the Wasserstein generative conditional sampler (WGCS) for CDE. GCDS and WGCS employ the idea of generative adversarial networks (GANs) \citep{NIPS2014_5ca3e9b1} to \textit{generate samples} from the conditional distribution $p(Y \mid \bm{X})$. In these deep generative models, random noise inputs are used to learn and generate samples from the target distribution. 
However, a well-known problem with deep generative networks is that the generator may simply memorize the training samples instead of generalizing to new test data \citep{Arpit2017ICML, vandenBurg2021NeurIPS}. When \textit{memorization} occurs, the random noise variable no longer reflects any variability, and the predicted density $p(Y \mid \bm{X})$ approaches a point mass. We refer to this phenomenon as \textit{vanishing variability}.

To conduct joint quantile estimation as well as conditional density estimation, we propose a new deep generative approach called penalized generative quantile regression (PGQR). PGQR employs a novel variability penalty to avoid vanishing variability. 
We further guarantee the monotonicity (i.e. non-crossing) of the estimated quantile functions from PGQR by designing a \textit{new} family of partial monotonic neural networks (PMNNs). The PMNN architecture ensures partial monotonicity with respect to quantile inputs, while retaining the expressiveness of neural networks. 

The performance of PGQR depends crucially on carefully choosing the variability penalty parameter $\lambda$, where $\lambda$ lies in the range of $[0, \lambda_{\max}]$ and $\lambda_{\max}$ should be determined before optimization. In all of our analyses, we chose $\lambda_{\max}=\exp(1)$. Unfortunately, common tuning procedures such as cross-validation are impractical for deep learning, since this would involve repetitive evaluations of the network for different training sets and choices of $\lambda$. Inspired by \cite{https://doi.org/10.48550/arxiv.2006.00767}, we construct our deep generative network in such a way that $\lambda$ is included as an additional 
random input. This way, only a \textit{single} optimization is needed to learn the neural network parameters, and then it is effortless to generate samples from a set of candidate values for $\lambda$. We propose a criterion for selecting the optimal choice of $\lambda$ to generate the final desired samples from $p(Y \mid \bm{X})$.

Our main contributions can be summarized as follows:
\begin{enumerate}
	\item We propose a deep generative approach for joint quantile regression and conditional density estimation called \textit{penalized generative quantile regression}. PGQR simultaneously generates samples from \textit{multiple} random quantile levels, thus precluding the need to refit the model at different quantiles.
	\item We introduce a novel variability penalty to avoid the vanishing variability phenomenon in deep generative models and apply this regularization technique to PGQR. 
	\item We construct a new family of partial monotonic neural networks (PMNNs) to circumvent the problem of crossing quantile curves.
	\item To facilitate scalable computation for PGQR and bypass computationally expensive cross-validation for tuning the penalty parameter, we devise a strategy that allows PGQR to be implemented using only a \textit{single} optimization.  
\end{enumerate}
The rest of the article is structured as follows. Section \ref{sec:PGQR} introduces the generative quantile regression framework and our variability penalty. In Section \ref{sec:mono}, we introduce the PMNN family for preventing quantile curves from crossing each other. Section \ref{sec:computation} discusses scalable computation for PGQR, namely how to tune the variability penalty parameter with only single-model training. In Section \ref{sec:sim}, we demonstrate the utility of PGQR through simulation studies and analyses of additional real datasets. Section \ref{sec:conclusion} concludes the paper with some discussion and directions for future research. All proofs of propositions can be found in Appendix \ref{sec:proofs}.

\section{Penalized Generative Quantile Regression} \label{sec:PGQR}
\subsection{Generative Quantile Regression} \label{sec:GQRnopenalty}

Before introducing PGQR, we first introduce our framework for generative quantile regression (without variability penalty). Given $n$ training samples $(\bm{X}_1, Y_1), (\bm{X}_2, Y_2), \ldots, (\bm{X}_n, Y_n)$ and quantile level $\tau \in (0,1)$, nonparametric quantile regression minimizes the empirical quantile loss,
\begin{equation} \label{nonparametricquantile}
	\argmin_{f \in \mathcal{F}} \frac{1}{n} \sum_{i=1}^n\rho_{\tau} \left( Y_i-f(\bm{X}_i, \tau) \right),
\end{equation}
where $\rho_{\tau}(u) = u(\tau-I(u<0))$ is the check function, and $\mathcal{F}$ is a function class such as a reproducing kernel Hilbert space or a family of neural networks. Neural networks are a particularly attractive way to model the quantile function $f(\bm{X}, \tau)$ in \eqref{nonparametricquantile}, because they are universal approximators for any Lebesgue integrable function \citep{lu2017expressive}. 

In this paper, we model the quantile function in \eqref{nonparametricquantile} using deep neural networks (DNNs). We define a DNN as a neural network with at least two hidden layers and a large number of hidden neurons. We refer to \cite{EmmertStreib2020FrontiersinAI} for a comprehensive review of DNNs. Despite the universal approximation properties of DNNs, we must \textit{also} take care to ensure that our estimated quantile functions do \textit{not} cross each other. Therefore, we have to consider a wide enough \textit{monotonic} function class $\cG^m$ to cover the true $Q_{Y \mid \bm{X}}(\cdot)$. We formally introduce this class $\cG^m$ in Section \ref{sec:mono}.

Let $G$ denote the constructed DNN, which is defined as a feature map function $\{ G \in \cG^m:\mathbb{R}^{p+1} \mapsto\mathbb{R}^1 \}$ that takes $(\bm{X}_i, \tau)$ as input and generates the conditional quantile for $Y_i$ given $\bm{X}_i$ at level $\tau$ as output. We refer to this generative framework as \emph{Generative Quantile Regression} (GQR). The optimization problem for GQR is
\begin{equation}\label{eq:GQR}
	\widehat{G} = \underset{G \in \cG^m}{\argmin} \  \frac{1}{n} \sum_{i=1}^n \mathbb{E}_{\tau} \big\{ \rho_{\tau} \big(Y_i - G(\bm{X}_i, \tau) \big) \big\},
\end{equation}
where $\widehat{G}$ denotes the estimated quantile function with optimized parameters (i.e. the weights and biases) of the DNN. Note that optimizing the integrative loss $\mathbb{E}_{\tau}\{ \cdot \}$ over $\tau$ in \eqref{eq:GQR} is justified by Proposition \ref{prop:PGQRexistence} introduced in the next section (i.e. we set $\lambda=0$ in Proposition \ref{prop:PGQRexistence}).

In order to solve \eqref{eq:GQR}, we can use stochastic gradient descent (SGD) with mini-batching \citep{EmmertStreib2020FrontiersinAI}. For each mini-batch evaluation, $\bm{X}_i$ is paired with a quantile level $\tau$ sampled from $\textrm{Uniform}(0,1)$. Consequently, if there are $M$ mini-batches, the expectation in \eqref{eq:GQR} can be approximated by a Monte Carlo average of the random $\tau$'s, i.e. we approximate $\mathbb{E}_{\tau} \big\{ \rho_{\tau} \big(Y_i - G(\bm{X}_i, \tau) \big) \big\}$ with $ M^{-1} \sum_{k=1}^{M} \big\{ \rho_{\tau_k} \big( Y_i - G(\bm{X}_i, \tau_k) \big) \big\}$. Once we have solved \eqref{eq:GQR}, it is straightforward to use $\widehat{G}$ to generate \textit{new} samples $\widehat{G}(\bm{X}, \xi_k), k =1, 2, \ldots, b$, at various quantile levels $\bm{\xi} = \{ \xi_1, \xi_2, \ldots, \xi_b \}$, where the $\xi$'s are random $\textrm{Uniform}(0,1)$ noise inputs. Provided that $b$ is large enough, the generated samples at $\bm{\xi} \in (0,1)^{b}$ can be used to reconstruct the full conditional density $p(Y \mid \bm{X})$.

As discussed in Section \ref{sec:intro}, there are several other deep generative models \citep{zhou2022deep, liu2021wasserstein} for generating samples from $p(Y \mid \bm{X})$. These generative approaches also take random noise $z$ as an input (typically $z \sim \mathcal{N}(0,1)$) to reflect variability, but there is no statistical meaning for $z$. In contrast, the random noise $\tau$ in GQR \eqref{eq:GQR} has a clear interpretation as a quantile level $\tau \in (0,1)$. Although GQR and these other deep generative approaches are promising approaches for conditional sampling, these methods are all unfortunately prone to memorization of the training data. When this occurs, the random noise does not generate \textit{any} variability. To remedy this, we now introduce our \textit{variability penalty} in conjunction with our GQR loss function \eqref{eq:GQR}.

\subsection{Variability Penalty for GQR}\label{sec:pen}
 Overparameterization in neural networks occurs when the number of learnable parameters (i.e. the weights and biases) is much greater than the number of training samples. In practice, it is common for a DNN to be overparameterized. Empirical and theoretical studies have shown that overparameterization improves the generalization and robustness of DNNs, since it greatly enhances the representation power of the DNN and simplifies the optimization landscape \citep{allen2019learning, zhang2021understanding, Soltanolkotabi2019, MontanariZhong2022}. Unfortunately, when a DNN is overparameterized to capture the underlying structure in the training data, the random noise in GQR is likely to reflect no variability, no matter what value it inputs. This is a common problem in deep generative models \citep{Arpit2017ICML, Arora2017, vandenBurg2021NeurIPS}. We refer this phenomenon as \textit{vanishing variability}. To be more specific, let $G(\cdot, z)$ be the generator function constructed by a DNN, where $z$ is a random noise variable following some reference distribution such as a standard Gaussian or a standard uniform. The vanishing variability phenomenon occurs when
\begin{equation}\label{eq:overfit}
	\widehat{G}(\bm{X}_i,z) = Y_i, \hspace{.2cm} i = 1,...,n.
\end{equation}
In other words, there is no variability when inputting different random noise $z$, generating almost surely a discrete point mass at the training data. Given a \textit{new} feature vector $\bm{X}_{\textrm{new}}$, $G(\bm{X}_{\textrm{new}}, z)$ can also only generate one novel sample from the true data distribution because of vanishing variability.

Since GQR takes the training data $\bm{X}\in \mathbb{R}^{p}$ as input and the target quantile estimate lies in $\mathbb{R}^1$, the weights in the DNN associated with $\bm{X}\in \mathbb{R}^{p}$ are very likely to overwhelm those associated with the noise input $\tau$. As a result, GQR is very prone to encountering vanishing variability, as are other generative approaches with multidimensional features. From another point of view, we can see that due to the nonnegativity of the check function, the GQR loss \eqref{eq:GQR} achieves a minimum value of zero when we have vanishing variability \eqref{eq:overfit}. 

To remedy this problem, we propose a new regularization term that encourages the network to have \textit{more} variability when vanishing variability occurs. Given a set of features $\bm{X} \in \mathbb{R}^{p}$, the proposed \textit{variability penalty} is formulated as follows:
\begin{equation}\label{eq:pen}
	\textrm{pen}_{\lambda, \alpha}(G(\bm{X},\tau),G(\bm{X},\tau^\prime)) = - \lambda \log\left \{ \Vert G(\bm{X},\tau) - G(\bm{X},\tau^\prime)\Vert_1 + 1/\alpha\right\},
\end{equation}
where $\lambda \geq 0$ is the hyperparameter controlling the degree of penalization, $\alpha>0$ is a fixed hyperparameter, and $\tau,\tau^\prime \sim \textrm{Uniform}(0,1)$. The addition of $1 / \alpha$ inside the logarithmic term of \eqref{eq:pen} is mainly to ensure that the penalty function is always well-defined (i.e. the quantity inside of $\log \{ \cdot \}$ of \eqref{eq:pen} cannot equal zero). Our sensitivity analysis in Appendix \ref{sec:furtheranalysis} shows that our method is not usually sensitive to the choice of $\alpha$. We find that fixing $\alpha=1$ or $\alpha=5$ works well in practice.

With the addition of the variability penalty \eqref{eq:pen} to our GQR loss function \eqref{eq:GQR}, our \textit{penalized} GQR (PGQR) method solves the optimization,
\begin{equation}\label{eq:PGQR}
	\widehat{G} = \underset{G \in \cG^{m}}{\argmin} \frac{1}{n} \sum_{i=1}^n\left[ \mathbb{E}_{\tau} \left\{ \rho_{\tau} \big(Y_i - G(\bm{X}_i, \tau) \big)\right\} + \mathbb{E}_{\tau, \tau^\prime} \left\{\text{pen}_{\lambda, \alpha}\left( G(\bm{X}_i,\tau), G(\bm{X}_i, \tau^\prime) \right) \right\}\right],
\end{equation}
where $\cG^{m}$ denotes the PMNN family introduced in Section \ref{sec:mono}, and $\tau$ and $\tau^\prime$ independently follow a uniform distribution on $(0,1)$. The expectation $\mathbb{E}_{\tau, \tau^\prime}$ is again approximated by a Monte Carlo average. Note that when $\lambda = 0$, the PGQR loss \eqref{eq:PGQR} reduces to the (non-penalized) GQR loss \eqref{eq:GQR}. 

The next proposition states that an estimated generator function $\widehat{G}(\bm{X}, \tau)$ under the PGQR objective \eqref{eq:PGQR} is equivalent to a neural network estimator $\widehat{g}_{\tau}(\bm{X})$ based on the individual (non-integrative) quantile loss function, provided that the family $\mathcal{G}^{m}$ in \eqref{eq:PGQR} is large enough. 

\begin{proposition}[equivalence of $\widehat g_\tau({\bm X})$ and $\widehat G({\bm X},\tau)$] \label{prop:PGQRexistence}
	For fixed $\lambda \geq 0$ and $\alpha > 0$, let $\widehat g_{\tau}(\bm{X}) =\argmin_{g \in \mathcal{H} } \sum_{i=1}^n  \left[\rho_\tau(Y_i-g(\bm{X}_i)) + \mathbb{E}_{\tau, \tau^\prime} \{ \text{\emph{pen}}_{\lambda, \alpha}(g(\bm{X}_i,\tau),g(\bm{X}_i,\tau^\prime)) \}\right]$,
	for a class of neural networks $\mathcal{H}$, and $ \tau,\tau^\prime\overset{\text{iid}}\sim \text{Uniform}(0,1)$. Consider a class of generator functions $\mathcal{G}$, where $\{ G \in \mathcal{G} :\mathbb{R}^p\times\mathbb{R}\times \mathbb{R
	}\mapsto \mathbb{R}\}$.
	Assume that for all $\bm{X} \in\mathcal{X} \subset \mathbb{R}^p$ and quantile levels $\tau \in (0,1)$, there exists $G \in\mathcal{G}$ such that the target neural network $\widehat g_{\tau}(\bm{X})$ can be represented by a hyper-network $G( \bm{X},\tau)$. Then, for $i=1,\dots,n$,
	$$
	\widehat g_{\tau}(\bm{X}_i)=\widehat G(\bm{X}_i,\tau) \: \text{ a.s.},
	$$
	with respect to the probability law related to  $\mathbb{E}_{\tau}$, where $\widehat{G}$ is a solution to \eqref{eq:PGQR}.
\end{proposition}


Proposition \ref{prop:PGQRexistence} justifies optimizing the \textit{integrative} quantile loss over $\tau$ in the PGQR objective \eqref{eq:PGQR}. The next proposition justifies adding the variability penalty \eqref{eq:pen} to the GQR loss \eqref{eq:GQR} by showing that there exists $\lambda>0$ so that memorization of the training data does \textit{not} occur under PGQR. 

\begin{proposition}[PGQR does not memorize the training data] \label{prop:nomemorization}
Suppose that $\alpha > 0$ is fixed in the variability penalty \eqref{eq:pen}. Denote any minimizer of the PGQR objective function \eqref{eq:PGQR} by $\widehat{G}$. Then, there exists $\lambda>0$ such that
\begin{align*}
	\widehat{G}(\bm{X}_i, \tau) \neq Y_i \hspace{.2cm} \textrm{for at least one } i =  1, \ldots, n.
\end{align*}
\end{proposition}


\begin{figure}[t]
\centering
\includegraphics[width=0.6\textwidth]{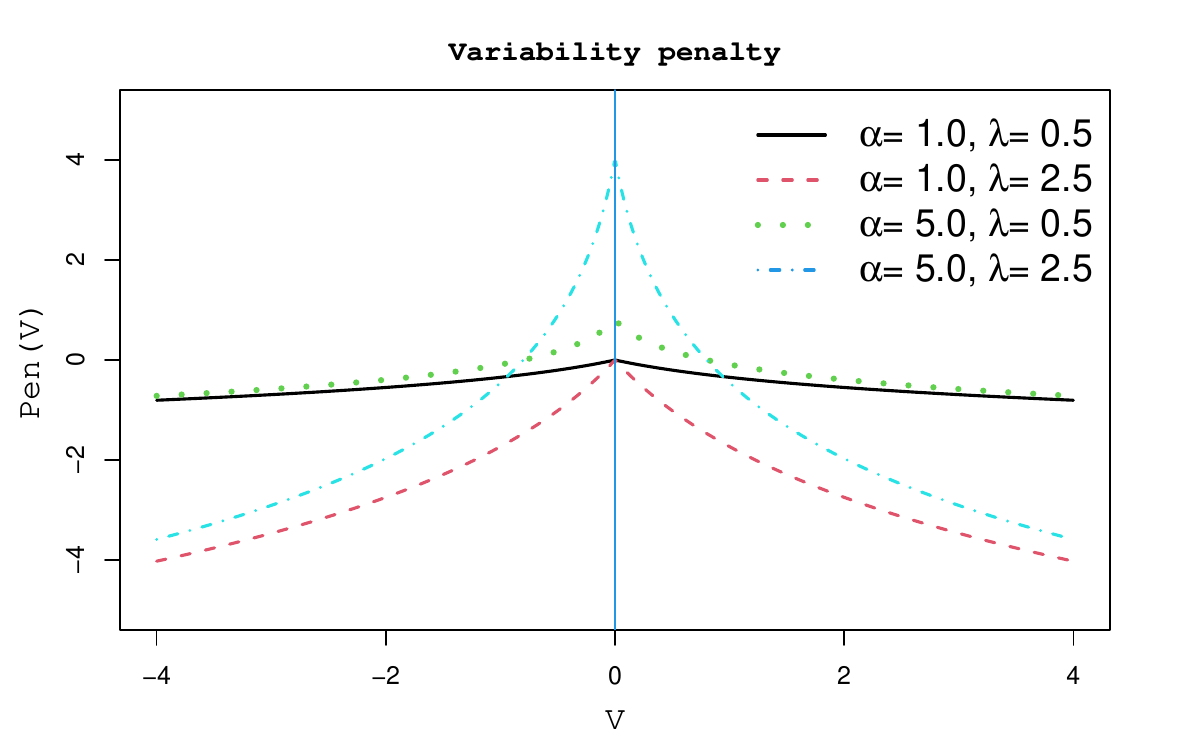}
\caption{\small A plot of the penalty function $\textrm{pen}_{\lambda, \alpha}(V) = -\lambda \log ( \lvert V \rvert + 1/\alpha) $.}
\label{fig:penalty}
\end{figure}

We now provide some intuition into the variability penalty \eqref{eq:pen}. We can quantify the amount of variability in the network by $V := G(\bm{X},\tau) - G(\bm{X},\tau^\prime)$. Clearly, $V=0$ when vanishing variability \eqref{eq:overfit} occurs. To avoid data memorization, we should thus penalize $V$ \textit{more heavily} when $V \approx 0$, with the maximum amount of penalization being applied when $V = 0$. Figure \ref{fig:penalty} plots the variability penalty, $\textrm{pen}_{\lambda, \alpha}(V) = -\lambda \log ( \lvert V \rvert + 1/\alpha)$, for several pairs of $(\alpha, \lambda)$. Figure \ref{fig:penalty} shows that $\textrm{pen}_{\lambda, \alpha}(V)$ is sharply peaked at zero and strictly convex decreasing in $\lvert V \rvert$. Thus, maximum penalization occurs when $V = 0$, and there is less penalization for larger $\lvert V \rvert$. As $\lambda$ increases, $\textrm{pen}_{\lambda, \alpha}(V)$ also increases for all values of $V \in (-\infty, \infty)$, indicating that larger values of $\lambda$ lead to more penalization.

Our penalty function takes a specific form \eqref{eq:pen}. However, any other penalty on $ G(\bm{X}, \tau) - G(\bm{X}, \tau^\prime)$ for $\bm{X} \in \mathcal X$ that has a similar shape as the PGQR penalty (i.e. sharply peaked around zero, as in Figure \ref{fig:penalty}) would also conceivably encourage the network to have greater variability. Another way to control the variability is to directly penalize the empirical variance $s^2 = (n-1)^{-1} \sum_{i=1}^{n} [G(\bm{X}_i, \tau) - \bar{G}]^2$, where $\bar{G} = n^{-1} \sum_{i=1}^{n} G(\bm{X}_i, \tau)$, or the empirical standard deviation $s = \sqrt{s^2}$. However, penalties on $s^2$ or $s$ do \textit{not} have an additive form and individual penalty terms depend on each other, and therefore, these are \textit{not} conducive to SGD-based optimization. 

We now pause to highlight the novelty of our constructed variability penalty. In many other nonparametric models, e.g. those based on nonparametric smoothing, a roughness penalty is often added to an empirical loss function in order to control the smoothness of the function or density estimate \citep{GreenSilverman1994}. These roughness penalties encourage greater smoothness of the target function in the spirit of reducing variance in the bias-variance trade-off. In contrast, our variability penalty \eqref{eq:pen} directly penalizes generators with \textit{too small} variance. By adding the penalty to the GQR loss \eqref{eq:GQR}, we are encouraging the estimated network to have \textit{more} variability.

\begin{figure}[t]
\centering
\includegraphics[width=1.0\textwidth]{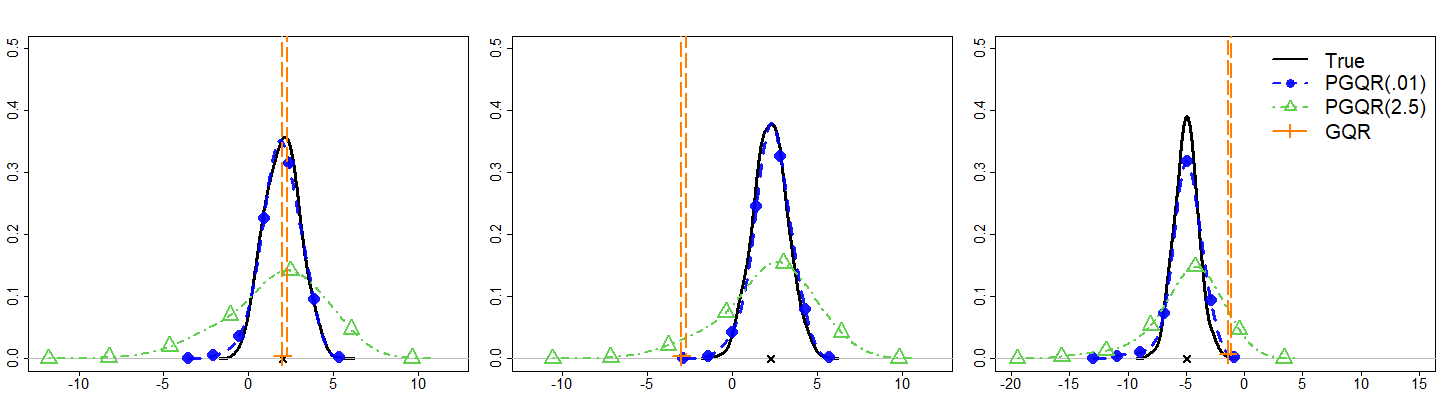}
\caption{\small Plots of the estimated conditional densities $p(Y \mid \bm{X}_{\text{test}})$ for three different test observations under non-penalized GQR and PGQR with two different choices of $\lambda \in \{ 0.01, 2.5 \}$. $\lambda^{\star} = 0.01$ is the optimal $\lambda$ chosen by our tuning parameter selection method in Section \ref{sec:select}.}
\label{fig:overfit}
\end{figure}

To illustrate how PGQR prevents vanishing variability, we carry out a simple simulation study under the model, $Y_i = \bm{X}_i^\top \boldsymbol{\beta} + \epsilon_i, i = 1, \ldots, n$,
where $\bm{X}_i \overset{iid}{\sim} \mathcal{N}(\boldsymbol{0},\bm{I}_{20})$, $\epsilon_i \overset{iid}{\sim} \mathcal{N}(0,1)$, and the coefficients in $\boldsymbol{\beta}$ are equispaced over $[-2, 2]$. We used 2000 samples to train GQR \eqref{eq:GQR} and PGQR \eqref{eq:PGQR}, and an additional 200 validation samples were used for tuning parameter selection of $\lambda$ (described in Section \ref{sec:select}). To train the generative network, we fixed $\alpha=1$ and tuned $\lambda$ from a set of 100 equispaced values between $0$ and $2.5$.  We then generated 1000 samples $\{ \widehat{G}(\bm{X}_{\text{test}},\xi_k,\lambda) \}_{k=1}^{1000}$, $\xi_k \overset{iid}{\sim} \textrm{Uniform}(0,1)$, from the estimated conditional density of $p(Y \mid \bm{X}_{\text{test}})$ for three different choices of out-of-sample test data for $\bm{X}_{\text{test}}$. In our simulation, the tuning parameter selection procedure introduced in Section \ref{sec:select} chose an optimal $\lambda$ of $\lambda^{\star} = 0.01$. 

As shown in Figure \ref{fig:overfit}, the \textit{non}-penalized GQR model introduced in Section \ref{sec:GQRnopenalty} (dashed orange line) suffers from vanishing variability. In particular, GQR is constructed by a neural network of three hidden layers, each layer having 1000 hidden neurons. Namely, GQR generates values near the true test sample $Y_{\text{test}}$ almost surely, despite the fact that we used 1000 different inputs for $\xi$ to generate the $\widehat{G}$ samples. In contrast, the PGQR model with optimal $\lambda^{\star} = 0.01$ (solid blue line with filled circles) approximates the true conditional density (solid black line) very closely, capturing the Gaussian shape and matching the true underlying variance. 

On the other hand, Figure \ref{fig:overfit} also illustrates that if $\lambda$ is chosen to be \textit{too} large in PGQR \eqref{eq:PGQR}, then the subsequent approximated conditional density will exhibit larger variance than it should. In particular, when $\lambda = 2.5$, Figure \ref{fig:overfit} shows that PGQR (dashed green line with hollow triangles) \textit{overestimates} the true conditional variance for $Y$ given $\bm{X}$. This demonstrates that the choice of $\lambda$ in the variability penalty \eqref{fig:penalty} plays a very crucial role in the practical performance of PGQR. In Section \ref{sec:select}, we describe how to select an ``optimal'' $\lambda$ so that PGQR neither underestimates \textit{nor} overestimates the true conditional variance. Our method for tuning $\lambda$ also requires only a \textit{single} optimization, making it an attractive and scalable alternative to cross-validation.

 As discussed earlier, the quantile loss function \eqref{eq:GQR} is equal to zero when the estimated quantiles are exactly equal to the observed responses $Y$'s. Therefore, if a complex model like a DNN is used for $G(\bm{X}, \tau)$ in \eqref{eq:GQR}, the estimated model is more prone to interpolate the observed responses. It is natural to wonder whether picking a simpler neural network (i.e. one with fewer hidden layers and/or neurons) will remedy the vanishing variability phenomenon. In Appendix \ref{sec:furtheranalysis}, we show that tuning the model complexity for \emph{non}-penalized GQR (e.g. choosing a simpler neural network with one or two hidden layers and five or 50 neurons) does \emph{not} completely avoid vanishing variability. A simpler neural network structure also invariably loses representation power. Our PGQR model \eqref{eq:PGQR} allows us to realize the benefits of overparameterization \citep{allen2019learning, zhang2021understanding, Soltanolkotabi2019, MontanariZhong2022} while \emph{simultaneously} avoiding vanishing variability.

\section{Partial Monotonic Neural Network}\label{sec:mono}
Without any constraints on the network, PGQR is just as prone as other unconstrained nonparametric quantile regression models to suffer from crossing quantiles. For a fixed data point $\bm{X}_i$ and an estimated network $\widehat{G}$, the \textit{crossing problem} occurs when 
\begin{equation}\label{eq:cross}
\widehat{G}(\bm{X}_i,\tau_1) > \widehat{G}(\bm{X}_i,\tau_2) \:\text{ when }\: 0<\tau_1<\tau_2<1.
\end{equation}
In the neural network literature, one popular way to address the crossing problem is to add a penalty term to the loss such as $\max(0, -\partial G(\bm{X}_i;\tau)/\partial \tau)$ \citep{tagasovska2018frequentist, liu2020certified, shen2022estimation}. Through regularization, the network is encouraged to have larger partial derivatives with respect to $\tau$. 
Another natural way to construct a monotonic neural network (MNN) is by restricting the weights of network to be nonnegative through a transformation or through weight clipping \citep{832655, daniels2010monotone,mikulincer2022size}. Because \textit{all} the weights are constrained to be nonnegative, this class of MNNs might require longer optimization (i.e. more epochs in SGD) and/or a large amount of hidden neurons to ensure the network's final expressiveness.

Instead of constraining \textit{all} the weights in the neural network, we make a simple modification to the MNN architecture which we call the \textit{partial} monotonic neural network (PMNN). The PMNN family consists of \textit{two} feedforward neural networks (FNN) as sub-networks which divide the input into two segments. The first sub-network is a weight-constrained quantile network $\mathbb{R}^1 \mapsto \mathbb{R}^h$ where the only input is the quantile level $\tau$ and $h$ denotes the number of hidden neurons. The FNN structure for one hidden layer in this sub-network is $g(\tau) = \sigma(\bm{U}_{\text{pos}} \tau+\bm{b})$, where $\bm{U}_{\text{pos}}$ is the weights matrix of \textit{nonnegative} weights, $\bm{b}$ is the bias matrix, and $\sigma$ is the activation function, such as the rectified linear unit (ReLU) function $\sigma(x) = \max \{0, x \}$. The $K_1$-layer constrained sub-network $g_c$ can then be constructed as $g_{c}(\tau) = g_K \circ \dots \circ g_1$.

The second sub-network is a $K_2$-layer \textit{unconstrained} network $g_{uc}: \mathbb{R}^p \mapsto \mathbb{R}^h$, taking the data $\bm{X}$ as inputs. The structure of one hidden layer in this sub-network is $g^{\prime}(\bm{X}) = \sigma(\bm{U} \bm{X} +\bm{b})$, where $\bm{U}$ is an \textit{unconstrained} weights matrix. We construct $g_{uc}$ as $g_{uc}(\bm{X}) = g_{K_2}^{\prime} \circ \dots \circ g_1^{\prime}$. Finally, we construct a single weight-constrained connection layer $f:\mathbb{R}^h \mapsto \mathbb{R}^1$ that connects the two sub-networks to quantile estimators,
\begin{equation}\label{eq:pmmn}
G(\bm{X}, \tau) = f \ \circ (g_c + g_{uc}).
\end{equation}
With our PMNN architecture, the monotonicity of $G(\cdot,\tau)$ with respect to $\tau$ is guaranteed by the positive weights in the quantile sub-network $g_c$ and the final layer $f$. At the same time, since the data sub-network $g_{uc}$ is unconstrained in its weights, $g_{uc}$ can learn the features of the training data well. In particular, the PMNN family is more flexible in its ability to learn features of the data and is easier to optimize than MNN because of this unconstrained sub-network. We use our proposed PMNN architecture as the family $\cG^{m}$ over which to optimize the PGQR objective for all of the simulation studies and real data analyses in this manuscript.

\begin{figure}[t]
\centering
\includegraphics[width=0.85\textwidth]{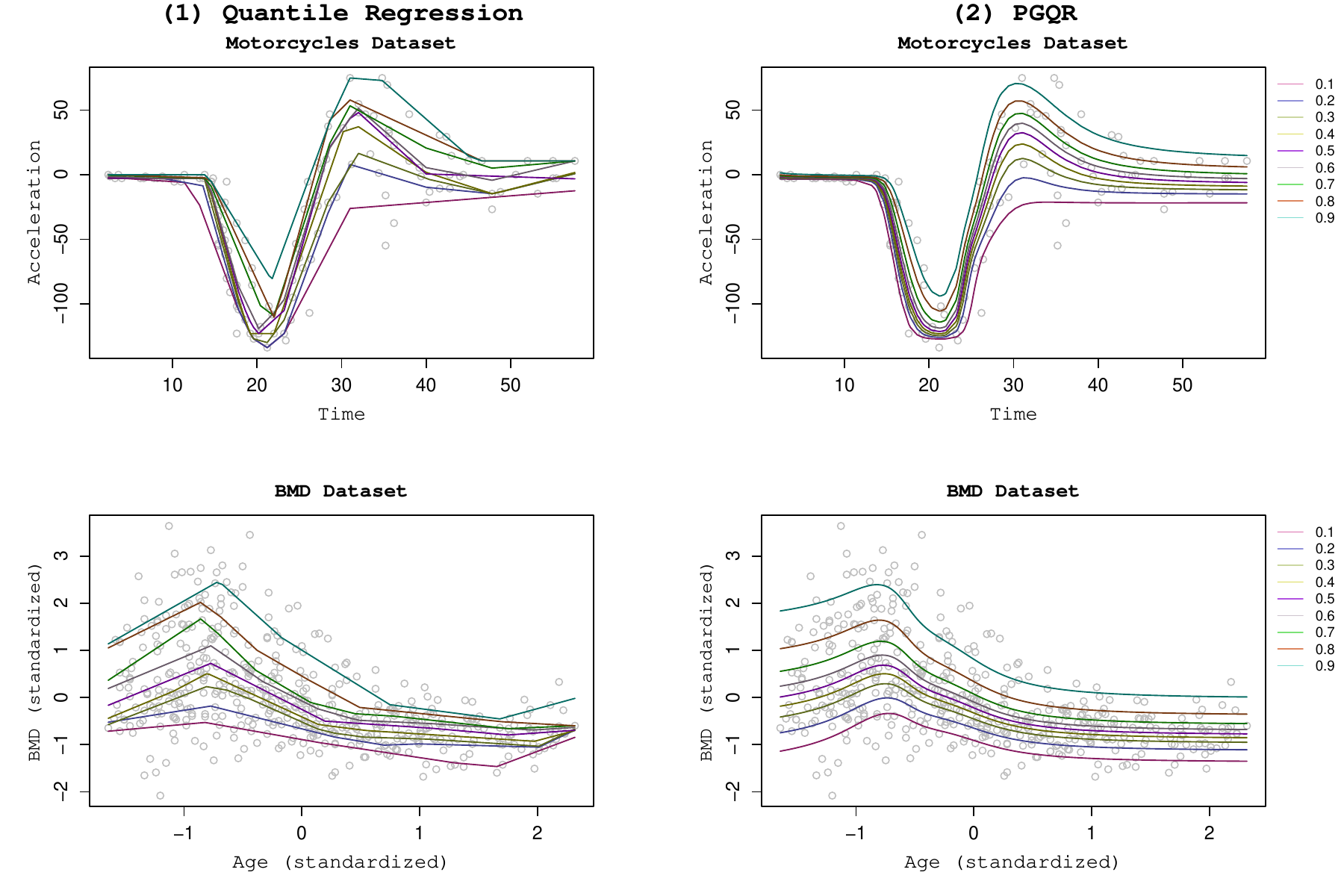}
\caption{\small Estimated quantile curves at levels $\tau \in \{ 0.1, 0.2, 0.3 , 0.4, 0.5, 0.6, 0.7, 0.8, 0.9\}$ for the motorcycle and BMD datasets. The left two panels plot the estimated curves for unconstrained nonparametric quantile regression, and the right two panels plot the estimated curves for PGQR with the PMNN family.}
\label{fig:crossdata}
\end{figure}

To investigate the performance of our proposed PMNN family of neural networks, we fit the PGQR model with PMNN as the function class $\cG^m$ on two benchmark datasets. The first dataset is the motorcycles data \citep{silverman1985some} where the response variable is head acceleration (in g) and the predictor is time from crash impact (in ms). The second application is a bone mineral density (BMD) dataset  \citep{takeuchi2006nonparametric} where the response is the standardized relative change in spinal BMD in adolescents and the predictor is the standardized age of the adolescents. For these two datasets, we estimated the quantile functions at different quantile levels $\tau \in \{ 0.1, 0.2, 0.3 , 0.4, 0.5, 0.6, 0.7, 0.8, 0.9\}$ over the domain of the predictor. We compared PGQR under PMNN to the unconstrained nonparametric quantile regression approach implemented in the \textsf{R} package \texttt{quantreg}. The left two panels of Figure \ref{fig:crossdata} demonstrate that the estimated quantile curves under \textit{unconstrained} nonparametric quantile regression are quite problematic for both the motorcycle and BMD datasets; the quantile curves cross each other at multiple points in the predictor domain. In comparison, the right two panels of Figure \ref{fig:crossdata} show that PGQR with the PMNN family ensures \textit{no} crossing quantiles for either dataset.


\section{Scalable Computation for PGQR}\label{sec:computation}

\subsection{Single-Model Training for PGQR} \label{sec:singlemodeltraining}
As illustrated in Section \ref{sec:pen}, PGQR's performance depends greatly on the regularization parameter $\lambda \geq 0$ in the variability penalty \eqref{eq:pen}. If $\lambda$ is too small or if $\lambda = 0$, then we may underestimate the true conditional variance of $p(Y \mid \bm{X})$ and/or encounter vanishing variability. On the other hand, if $\lambda$ is too large, then we may overestimate the conditional variance. Due to the large number of parameters in a DNN, tuning $\lambda$ would be quite burdensome if we had to repeatedly evaluate the deep generative network $G(\cdot,\tau)$ for multiple choices of $\lambda$ and training sets.  In particular, using cross-validation to tune $\lambda$ is infeasible for DNNs, especially if the size of the training set is large. 

Inspired by the idea of \emph{Generative Multiple-purpose Sampler} (GMS) \citep{https://doi.org/10.48550/arxiv.2006.00767}, we instead propose to \textit{include} $\lambda$ as an additional input in the generator, along with data $\bm{X}$ and quantile $\tau$. In other words, we impose a discrete uniform distribution $p(\lambda)$ for $\lambda$ whose support $\Lambda$ is a grid of candidate values for $\lambda$. We then estimate the network $G(\cdot,\tau,\lambda)$ with the modified PGQR loss function,
\begin{equation}\label{eq:PGQR2}
\widehat{G} = \underset{G \in \cG^m}{\argmin} \  \frac{1}{n} \sum_{i=1}^n \mathbb{E}_{\tau, \lambda} \bigg[ \rho_{\tau} \big(Y_i - G(\bm{X}_i, \tau, \lambda) \big) + \mathbb{E}_{\tau, \tau'} \left\{ \text{pen}_{\lambda, \alpha} \left( G(\bm{X}_i, \tau, \lambda), G(\bm{X}_i, \tau^\prime, \lambda)  \right) \right\} \bigg].
\end{equation}
Note that \eqref{eq:PGQR2} differs from the earlier formulation \eqref{eq:PGQR} in that $\lambda \in \Lambda$ is \textit{not} fixed in advance. In the PMNN family $\cG^m$, $\lambda$ is included in the unconstrained sub-network $g_{uc}: \mathbb{R}^{p+1} \mapsto \mathbb{R}^h$. 

The next corollary to Proposition \ref{prop:PGQRexistence} justifies including $\lambda$ in the generator and optimizing the integrative loss over $\{ \tau, \lambda \}$ in the modified PGQR objective \eqref{eq:PGQR2}. The proof of Corollary \ref{corollary:existence} is a straightforward extension of the proof of Proposition \ref{prop:PGQRexistence} and is therefore omitted. 

\begin{corollary} \label{corollary:existence} Let $\widehat g_{\tau,\lambda}(\bm{X})=\argmin_{g \in \mathcal{H} } \sum_{i=1}^n\left[ \rho_\tau(Y_i-g(\bm{X}_i)) + \mathbb{E}_{\tau, \tau'}  \{
\textrm{pen}_{\lambda, \alpha} \left( g(\bm{X}_i, \tau), g(\bm{X}_i, \tau^\prime ) \right) \} \right],$ for a class of neural networks $\mathcal{H}$, where $\alpha > 0$ is fixed and $\tau, \tau' \overset{iid}{\sim} \textrm{Uniform}(0,1)$. Consider a class of generator functions $\mathcal{G}$ where $\{ G \in \mathcal{G} : \mathbb{R}^p\times\mathbb{R}\times \mathbb{R
}\mapsto \mathbb{R}\}$.
Suppose that for all $\bm{X} \in\mathcal{X}\subset \mathbb{R}^p$, quantile levels $\tau \in (0,1)$, and tuning parameters $\lambda\in \Lambda$, there exists $G \in\mathcal{G}$ such that the target neural networks $\widehat g_{\tau,\lambda}(\bm{X})$  can be represented by some $G(\bm{X},\tau,\lambda)$. Then, for $i=1,\dots,n$,
$$
\widehat g_{\tau,\lambda}(\bm{X}_i)=\widehat G(\bm{X}_i,\tau,\lambda) \: \text{ a.s.},
$$
with respect to the probability law related to  $\mathbb{E}_{\tau, \lambda}$, where $\widehat{G}$ is a solution to \eqref{eq:PGQR2}.
\end{corollary}

We note that \cite{https://doi.org/10.48550/arxiv.2006.00767} proposed to use GMS for inference in \textit{linear} quantile regression. In contrast, PGQR is a method for \textit{nonparametric} joint quantile estimation and conditional density estimation (CDE). Linear quantile regression cannot be used for CDE; as a linear model, GMS linear quantile regression also does not suffer from vanishing variability. This is not the case for GQR, which motivates us to introduce the variability penalty in Section \ref{sec:pen}. Finally, we employ the GMS idea for \textit{tuning parameter} selection in PGQR, rather than for constructing pointwise confidence bands (as in \cite{https://doi.org/10.48550/arxiv.2006.00767}).

By including $\lambda \sim p(\lambda)$ in the generator, PGQR only needs to perform one \textit{single} optimization in order to estimate the network $G$. We can then \textit{use} the estimated network $\widehat{G}$ to tune $\lambda$, as we detail in the next section. This single-model training stands in contrast to traditional smoothing methods for nonparametric quantile regression, which typically require \textit{repeated} model evaluations via (generalized) cross-validation to tune hyperparameters such as bandwidth or roughness penalty.

\subsection{Selecting An Optimal Regularization Parameter} \label{sec:select}

We now introduce our method for selecting the regularization parameter $\lambda$ in our variability penalty \eqref{eq:pen}. Considering the candidate set $\Lambda$, the basic idea is to select the best $\lambda^\star \in \Lambda$ that minimizes the distance between the generated conditional distribution and the true conditional distribution. 
Denote the trained network under the modified PGQR objective \eqref{eq:PGQR2} as $\widehat{G}(\cdot,\tau,\lambda)$, and let $(\bm{X}_{\text{val}}, Y_{\text{val}})$ denote a validation sample. After optimizing \eqref{eq:PGQR2}, it is effortless to generate the conditional quantile functions for each $\bm{X}_{\text{val}}$ in the validation set, each $\lambda \in \Lambda$, and any quantile level $\xi \in (0,1)$. In order to select the optimal $\lambda^\star \in \Lambda$, we first introduce the following proposition.

\begin{proposition}\label{prop:lambda}
Suppose that two univariate random variables $Q$ and $W$ are independent. Then, $Q$ and $W$ have the same distribution if and only if $P(Q<W \mid W)$ follows a standard uniform distribution.  
\end{proposition}

Based on Proposition \ref{prop:lambda}, we know that, given a validation sample $(\bm{X}_{\text{val}}, Y_{\text{val}})$, the best $\lambda^\star \in \Lambda$ should satisfy $P_{\tau, \lambda} ( \widehat{G}(\bm{X}_{\text{val}}, \tau, \lambda^\star ) < Y_{\text{val}} \hspace{.1cm} \big\lvert \hspace{.1cm} Y_{\text{val}}) \sim \text{Uniform}(0,1).$
In practice, if we have $M$ random quantile levels $\bm{\tau} = (\tau_1, \ldots, \tau_M) \in (0,1)^{M}$, we can estimate the probability $P_{\tau, \lambda} := P_{\tau, \lambda} ( \widehat{G}(\bm{X}_{\text{val}}, \tau, \lambda^\star ) < Y_{\text{val}} \hspace{.1cm} \big\lvert \hspace{.1cm} Y_{\text{val}})$ with a Monte Carlo approximation,
\begin{equation} \label{eq:phattau}
\widehat{P}_{\tau, \lambda}^{(i)} = M^{-1} \sum_{k=1}^{M} \mathbb{I} \{ \widehat{G} (\bm{X}_{\text{val}}^{(i)},\hspace{.1cm} \tau_k, \lambda ) < Y_{\text{val}}^{(i)} \}\approx P_{\tau, \lambda} ( \widehat{G}(\bm{X}_{\text{val}}, \tau, \lambda^\star ) < Y_{\text{val}}\big\lvert \hspace{.1cm} Y_{\text{val}}),
\end{equation}
for each $i$th validation sample $(\bm{X}_{\text{val}}^{(i)}, Y_{\text{val}}^{(i)})$. With $n_{\text{val}}$ validation samples, we proceed to estimate $\widehat{P}_{\tau, \lambda}^{(i)}$ for all $n_{\text{val}}$ validation samples $(\bm{X}_{\text{val}}^{(i)}, Y_{\text{val}}^{(i)})$'s and each $\lambda \in \Lambda$. We then compare the empirical distribution of the $\widehat{P}_{\tau, \lambda}^{(i)}$'s to a standard uniform distribution for each $\lambda \in \Lambda$ and select the $\lambda^\star$ that \textit{minimizes} the distance between $\widehat{P}_{\tau, \lambda}$ and $\text{Uniform}(0,1)$. That is, given a valid distance measure $d$, we select the $\lambda^{\star}$ which satisfies 
\begin{equation} \label{eq:optlambda}
\lambda^{\star} =  \argmin_{\lambda \in \Lambda} d ( \widehat{P}_{\tau, \lambda}, \textrm{Uniform}(0,1) ).
\end{equation}

\begin{figure}[t]
\centering
\includegraphics[width=0.45\textwidth]{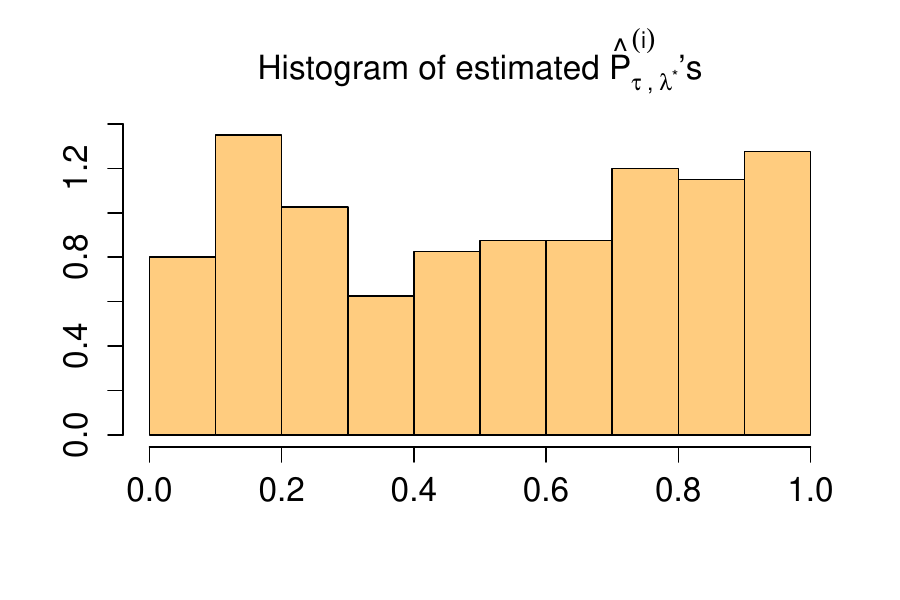}
\caption{\small Histogram of the estimated $\widehat{P}_{\tau, \lambda^{\star}}^{(i)}$'s in the validation set for the optimal $\lambda^{\star}$.}
\label{fig:simu-u}
\end{figure}

\noindent There are many possible choices for $d$ in \eqref{eq:optlambda}. In this paper, we use the Cramer?von Mises (CvM) criterion due to the simplicity of its computation and its good empirical performance. The CvM criterion is straightforward to compute as
\begin{align} \label{eq:CvM}
d ( \widehat{P}_{\tau, \lambda}, \textrm{Uniform}(0,1) ) =  \sum_{i=1}^{n_\text{val}} ( i/n_\text{val}-\widehat{P}_{\tau, \lambda}^{(i)} - 2/{n_\text{val}})^2/n_\text{val}.
\end{align}
After selecting $\lambda^*$ according to \eqref{eq:optlambda}, we can easily generate samples from the conditional quantiles of $p( Y \mid \bm{X})$ for any feature data vector $\bm{X}$. We simply generate $\widehat{G}(\bm{X}, \xi_k, \lambda^{\star})$, where $k = 1, \ldots, b$, for random quantiles $\xi_1, \ldots, \xi_b \overset{iid}{\sim} \textrm{Uniform}(0,1)$. For sufficiently large $b$, the conditional density $p(Y \mid \bm{X})$ can be inferred from $\{ \widehat{G}(\bm{X}, \xi_k, \lambda^{\star}) \}_{k=1}^{b}$. The complete algorithm for scalable computation of PGQR is given in Algorithm \ref{alg:PGQR}. 

We now illustrate our proposed selection method for $\lambda \in \Lambda$ in the simulated linear regression example from Section \ref{sec:pen}. Figure \ref{fig:simu-u} plots the histogram of the estimated $\widehat{P}_{\tau, \lambda^{\star}}^{(i)}$'s in the validation set for the $\lambda^{\star}$ which minimizes the CvM criterion. We see that the empirical distribution of the $\widehat{P}_{\tau, \lambda^{\star}}^{(i)}$'s closely follows a standard uniform distribution. 

Figure \ref{fig:cvm} illustrates how the PGQR solution changes as $\log(\lambda)$ increases for every $\lambda \in \Lambda$. Figure \ref{fig:cvm} plots the conditional standard deviation (left panel), the CvM criterion (middle panel), and the coverage rate of the 95\% confidence intervals (right panel) for the PGQR samples in the validation set. The vertical dashed red line denotes the optimal $\lambda^{\star} = 0.01$. We see that the $\lambda^{\star}$ which minimizes the CvM (middle panel) captures the true conditional standard deviation of $\sigma = 1$ (left panel) and attains coverage probability close to the nominal rate (right panel). This example demonstrates that our proposed selection method provides a practical alternative to computationally burdensome cross-validation for tuning $\lambda$ in the variability penalty \eqref{eq:pen}. In Appendix \ref{sec:sims}, we further demonstrate that our tuning parameter selection method is suitable to use even when the true conditional variance is very small (true $\sigma^2=0.01$). In this scenario, our procedure selects a very small $\lambda^{\star} \approx 0$, so that PGQR does \textit{not} overestimate the conditional variance.

\begin{figure}[t]
	\centering
	\includegraphics[width=1.0\textwidth]{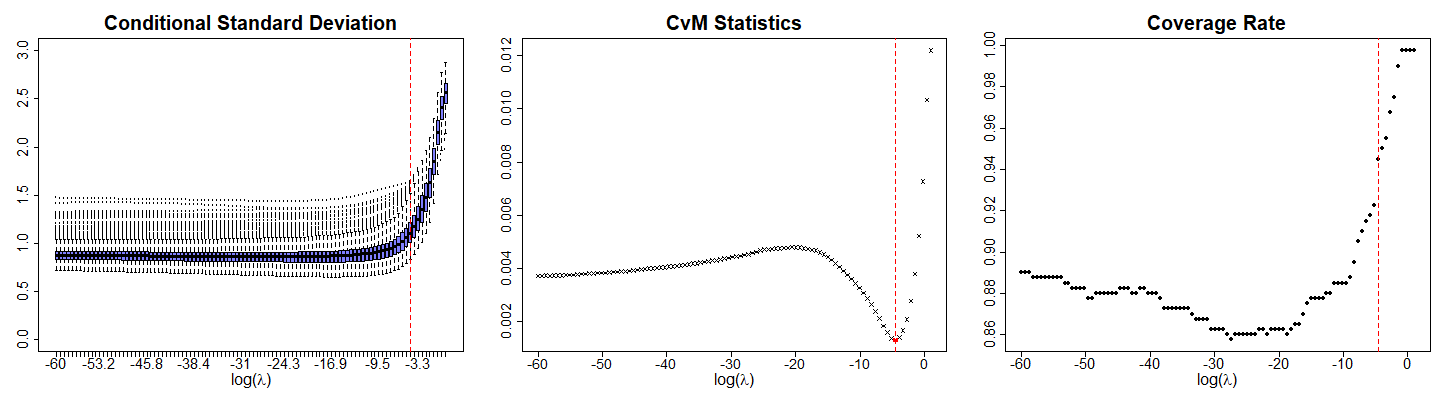}
	\caption{\small The conditional standard deviation (left panel), CvM statistic (middle panel), and coverage rate (right panel) in the validation set vs. $\log(\lambda)$ for each $\lambda \in \Lambda$. The red dashed line corresponds to the $\lambda^\star$ which minimizes the CvM criterion.}
	\label{fig:cvm}
\end{figure}


\begin{algorithm}[t]
	\footnotesize
	\caption{\footnotesize Scalable Implementation of PGQR}\label{alg:PGQR}
	\begin{algorithmic}[1]
		\STATE \emph{Split} \textit{non}-test data into non-overlapping training set $(\bm{X}_{\text{train}}, \bm{Y}_{\text{train}})$ and validation set $(\bm{X}_{\text{val}}, \bm{Y}_{\text{val}})$
		
		\vspace{-.02cm}
		\STATE \emph{Initialize} $G$ parameters $\phi$, learning rate $\gamma$, width $h$ of DNN hidden layers, and total epoches $T$
		\STATE \textbf{procedure:} Optimizing $G$ 
		\begin{ALC@g}
			\FOR{epoch $t$ in $1,\ldots, T$}
			\STATE Sample $\tau,\tau^\prime \sim \textrm{Uniform}(0,1)$ and $\lambda \in \Lambda$
			\STATE Evaluate loss \eqref{eq:PGQR2} with $(\bm{X}_{\text{train}}, \bm{Y}_{\text{train}})$
			\STATE Update $G$ parameters $\phi$ via SGD
			\ENDFOR
			\STATE \textbf{return} $\widehat{G}$
		\end{ALC@g}
		\STATE \textbf{end procedure}
		
		\STATE \textbf{procedure:} Tuning $\lambda$ 
		\begin{ALC@g}
			\STATE Set $M = 1000$ and sample $\tau_1, \ldots, \tau_M \overset{iid}{\sim} \textrm{Uniform}(0,1)$
			\STATE Generate $\{ \widehat{G}(\bm{X}_{\text{val}}^{(i)}, \tau_1, \lambda_l ), \ldots, \widehat{G}(\bm{X}_{\text{val}}^{(i)}, \tau_M, \lambda_l ) \}$ for each $\lambda_l \in \Lambda$ and each $\bm{X}_{\text{val}}^{(i)}$, $i = 1, \ldots, n_{\text{val}}$ 
			\STATE Compute $\widehat{P}_{\tau, \lambda_l}^{(i)}$ as in \eqref{eq:phattau} on $(\bm{X}_{\text{val}}^{(i)}, Y_{\text{val}}^{(i)})$ for each $i = 1, \ldots, n_{\text{val}}$ and each $\lambda_l \in \Lambda$
			\STATE Select $\lambda^*$ according to \eqref{eq:optlambda} with CvM criterion \eqref{eq:CvM} as $d$
			\STATE \textbf{return} $\lambda^*$
		\end{ALC@g}
		\STATE \textbf{end procedure}
		
		\STATE \textbf{procedure:} Estimating $p(Y \mid \bm{X}_{\text{test}})$ for test data $\bm{X}_{\text{test}}$
		\begin{ALC@g}
			\STATE Set $b=1000$ and sample $\xi_1, \ldots, \xi_b \overset{iid}{\sim} \text{Uniform}(0,1)$
			\STATE Estimate $\xi_k$-th quantile of $Y \mid \bm{X}_{\textrm{test}}$ as $\widehat{G}(\bm{X}_{\text{test}},\xi_k,\lambda^*)$ for $k = 1, \ldots, b$
			\STATE \textbf{return} $\{\widehat{G}(\bm{X}_{\text{test}},\xi_1,\lambda^*), \ldots, \widehat{G}(\bm{X}_{\text{test}},\xi_b,\lambda^*)\}$
		\end{ALC@g}
		\STATE \textbf{end procedure}
	\end{algorithmic}
\end{algorithm}



\section{Numerical Experiments and Real Data Analysis}\label{sec:sim}

	We evaluated the performance of PGQR on several simulated and real datasets. We fixed $\alpha=1$ or $\alpha = 5$ in the variability penalty \eqref{eq:pen}. We optimized the modified PGQR loss \eqref{eq:PGQR2} over the PMNN family (Section \ref{sec:mono}), where each sub-network had three hidden layers with 1000 neurons per layer and ReLU \citep{nair2010rectified} was used as the activation function. The support for $\Lambda$ was chosen to be 100 equispaced values between $0$ and $\exp(1)$. \texttt{PyTorch} was used to implement PGQR. To ensure numerical stability, Algorithm \ref{alg:PGQR} was initialized with random weights close to but not exactly zero. We estimated $\widehat{G}$ in \eqref{eq:PGQR2} using the Adam optimizer \citep{kingma2014adam}, which has been empirically shown to be robust and helps to mitigate gradient explosion. Gradient clipping \citep{pmlr-v28-pascanu13} could also be employed to prevent exploding gradients; however, we found that gradient clipping was not needed for PGQR. Throughout our simulations and real data applications, we did not experience any numerical overflow, underflow, or exploding gradients.

 We compared PGQR to several other state-of-the-art methods:
\begin{itemize}[leftmargin=.2in]
	\setlength\itemsep{0.2em}    
	\item \noindent{\bf GCDS} \citep{zhou2022deep}. Following \cite{zhou2022deep}, we trained both a generator and a discriminator using FNNs with one hidden layer of 25 neurons. Despite this simple architecture, we found that GCDS might still encounter vanishing variability. For fair comparison to PGQR, we also increased the number of hidden layers to three and the number of nodes per layer to 1000. We refer to this modification as \textbf{deep-GCDS}.
	\item \noindent{\bf WGCS} \citep{liu2021wasserstein}. We adopted the gradient penalty recommended by \cite{liu2021wasserstein} and set the hyperparameter associated with the gradient penalty as 0.1.
	
	\item \noindent{\bf FlexCoDE}  \citep{izbicki2017converting}. We considered three ways of estimating the basis functions $\beta_j(\bm{X})$: nearest-neighbor regression (NNR), sparse additive model (SAM), and XGBoost (FlexZboost). 
	\item \noindent{\bf Random Forest CDE}, or RFCDE \citep{https://doi.org/10.48550/arxiv.1906.07177}. 
\end{itemize}
For FlexCoDE and RFCDE, we adopted the default hyperparameter settings of \cite{izbicki2017converting} and \cite{https://doi.org/10.48550/arxiv.1906.07177} respectively.

\begin{figure}[t]
	\centering
	\includegraphics[width=1.0\textwidth]{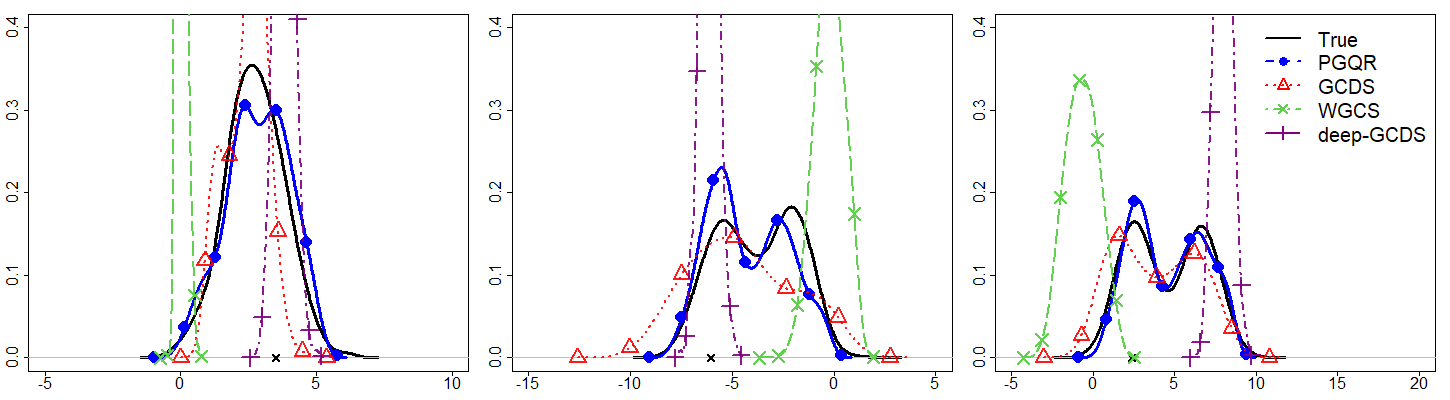}
	\caption{\small Plots of the estimated conditional densities $p(Y \mid \bm{X}_{\text{test}})$ for three different test observations from one replication of Simulation 3. Plotted are the estimated conditional densities for PGQR ($\alpha=1$), GCDS, WGCS, and deep-GCDS.}
	\label{fig:simu}
\end{figure}

\subsection{Simulation Studies}\label{subsec:construction}


For our simulation studies, we generated data from $Y_i = g(\bm{X}_i) + \epsilon_i, i = 1, \ldots, 2000$, for some function $g$, where $\bm{X}_i \overset{iid}{\sim} \mathcal{N} (\bm{0}, \bm{I}_p)$ and the residual errors $\epsilon_i$'s were independent. The specific simulation settings are described below.
\begin{enumerate}[leftmargin=.2in]
	\setlength\itemsep{0.2em}
	
	\item \noindent{\bf Simulation 1: Multimodal and heteroscedastic.} $Y_i = \beta_i X_i + \epsilon_i$, where $\beta_i = \{-1,0,1\}$ with equal probability and $\epsilon_i = (0.25 \cdot \vert X_i \vert)^{1/2}$.
	
	\item \noindent{\bf Simulation 2: Mixture of left-skewed and right-skewed.} $Y_i = \bm{X}_i^{\top} \boldsymbol{\beta} + \epsilon_i$, where $\boldsymbol{\beta} \in \mathbb{R}^5$ is equispaced between $[-2, 2]$, and $\epsilon_i = \chi^2(1,1)\mathcal{I}(X_1>=0.5)+\text{log}[\chi^2(1,1)]\mathcal{I}(X_1<0.5)$. Here, the skewness is controlled by the covariate $X_1$.
	
	\item \noindent{\bf Simulation 3: Mixture of unimodal and bimodal.} $Y_i =\bm{X}_i^{\top}\boldsymbol{\beta}_i + \epsilon_i$, where $\bm{\beta}_i \in \mathbb{R}^5$ with $\beta_{i1} \in \{ -2, 2 \}$ with equal probability, $(\beta_{i2}, \beta_{i3}, \beta_{i4}, \beta_{i5})^\top$ are equispaced between $[-2, 2]$, and $\epsilon_i \overset{iid}{\sim} \mathcal{N}(0,1)$. Here, $p(Y \mid \bm{X})$ is unimodal when $X_1 \approx 0$, and otherwise, it is bimodal.
\end{enumerate}
In our simulations, 80\% of the data was used to train the model, 10\% was used as the validation set for tuning parameter selection, and the remaining 10\% was used as test data to evaluate model performance. In Appendix \ref{sec:sims}, we present additional simulation results for the following scenarios: $g(\bm{X}_i)$ is a nonlinear function of $\bm{X}_i$ (Simulation 4), the error variance is very small (Simulation 5), and the error variance is dependent on $\Vert \bm{X} \Vert_1$ (Simulation 6).

Figure \ref{fig:simu} compares the estimated conditional density of $p(Y \mid \bm{X}_{\text{test}})$ for three test observations from one replication of Simulation 3. We see that PGQR (solid blue line with filled circles) is able to capture both the unimodality \textit{and} the bimodality of the ground truth conditional densities (solid black line). Meanwhile, GCDS (dashed red line with hollow triangles), WGCS (dashed green line with crosses), and deep-GCDS (dashed purple line with pluses) struggled to capture the true conditional densities for at least some test points. In particular, Figure \ref{fig:simu} shows some evidence of variance \textit{underestimation} for WGCS and deep-GCDS, whereas this is counteracted by the variability penalty in PGQR. Additional figures from our simulation studies are provided in Appendix \ref{sec:sims}.

\begin{table}[t]
	\centering
	\resizebox{1.0\columnwidth}{!}{
		\begin{tabular}{c|ccc|ccc|ccc}
			\hline
			&\multicolumn{3}{c|}{Simulation 1}&\multicolumn{3}{c|}{Simulation 2}&\multicolumn{3}{c}{Simulation 3}\\
			\hline
			Method&$\mathbb{E}(Y \mid \bm{X})$&$\text{sd}(Y \mid \bm{X})$&Cov (Width)&$\mathbb{E}(Y \mid \bm{X})$&$\text{sd}(Y \mid \bm{X})$&Cov (Width)&$\mathbb{E}(Y \mid \bm{X})$&$\text{sd}(Y \mid \bm{X})$&Cov (Width)\\
			\hline
			\hline
			PGQR ($\alpha$=1)
			&0.41&0.34&0.95 (23.48)&0.38&0.11&0.93 (8.14)&0.30&0.08&0.92 (6.61)
			\vspace{0.05cm}\\
			PGQR ($\alpha$=5)
			&\textbf{0.36}&\textbf{0.31}&0.95 (23.41)&\textbf{0.31}&\textbf{0.07}&\textbf{0.95 (8.83)} &\textbf{0.25}&\textbf{0.06}&\textbf{0.96 (6.60)}
			\vspace{0.05cm}\\
			GCDS
			&10.49&25.82&0.68 (15.48)&0.53&0.27&0.92 (8.55)&0.33&0.12&0.84 (5.78)
			\vspace{0.05cm}\\
			WGCS
			&229.91&73.68&0.15 (4.88)&6.57&1.45&0.80 (9.17)&6.25&1.98&0.71 (8.08)
			\vspace{0.05cm}\\
			deep-GCDS
			&7.99&56.87&0.38 (7.42)&5.41&2.89&0.42 (2.12)&6.11&2.78&0.28 (2.04)
			\vspace{0.05cm}\\
			\hline
			\hline
			FlexCoDE-NNR
			&0.97&0.54&0.96 (23.94)&0.83&0.36&0.91 (9.03)&1.12&0.75&0.92 (9.11)
			\vspace{0.05cm}\\
			FlexCoDE-SAM
			&0.37&0.62&0.97 (25.07)&0.73&1.03&0.93 (11.15)&1.01&1.99&0.93 (10.91)
			\vspace{0.05cm}\\
			FlexZBoost
			&0.77&63.81&\textbf{1.00 (46.11)}&1.29&0.36&0.91 (8.28)&1.77&0.74&0.85 (7.88)
			\vspace{0.05cm}\\
			RFCDE
			&1.98&0.72&0.44 (25.46)&0.61&0.34&0.56 (6.26)&0.83&0.65&0.96 (23.06)
			\\
			\hline
	\end{tabular}}
	\vspace{-0.2cm}
	\caption{\small Table reporting the PMSE for the conditional expectation and standard deviation, as well as the coverage rate (Cov) and average width of the 95\% prediction intervals, for Simulations 1 through 3. Results were averaged across 20 replicates.}
	\label{tab:simulation-main}
\end{table}

We repeated our simulations for 20 replications. For each experiment, we computed the predicted mean squared error (PMSE) for different summaries of the conditional densities for the test data. We define the PMSE as
\begin{equation}\label{eq:PMSE}
	\text{PMSE}= \frac{1}{n_{\text{test}}} \sum_{i=1}^{n_{\text{test}}}\left( \widehat{m} (\bm{X}_{\text{test},i})- m(\bm{X}_{\text{test},i}) \right)^2,
\end{equation}
where $m(\bm{X})$ generically refers to the conditional mean of $Y$ given $\bm X$, i.e. $\mathbb{E}( Y \mid \bm{X})$, or the conditional standard deviation, i.e. $\text{sd}( Y \mid \bm{X})$. For the generative models (PGQR, GCDS, WGCS, and deep-GCDS), we approximated $\widehat{m}$ in \eqref{eq:PMSE} by Monte Carlo simulation using $1000$ generated samples, while for FlexCoDE and RFCDE, we approximated $\widehat{m}$ using numerical integration. In addition to PMSE, we also used the 2.5th and 97.5th percentiles of the predicted conditional densities to construct 95\% prediction intervals for the test data. We then calculated the coverage probability (Cov) and the average width for these prediction intervals. 


\begin{figure}[t]
	\centering
	\includegraphics[width=\textwidth]{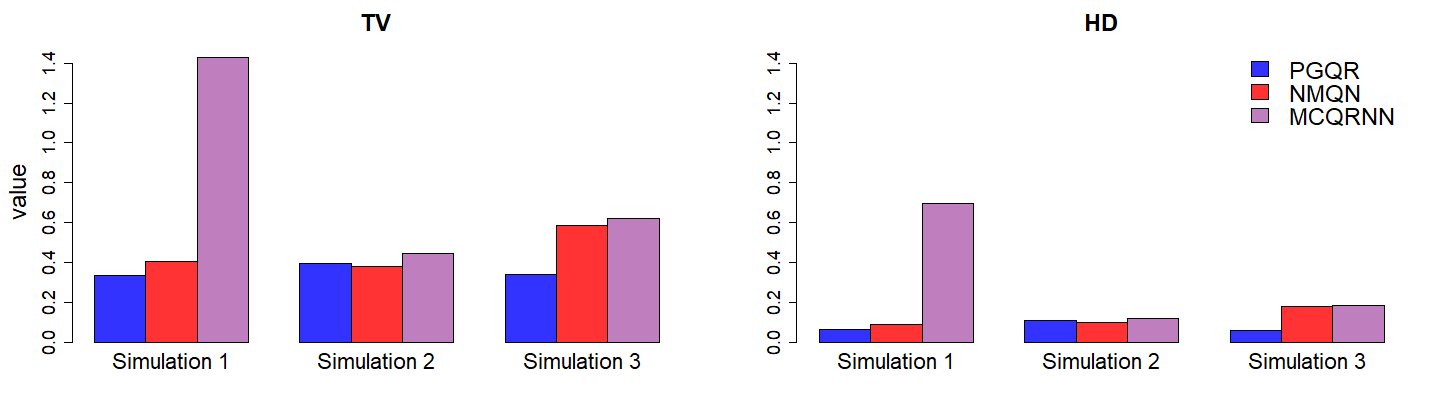}
	\caption{\small Barplots of the average total variation distance (TV) and Hellinger distance (HD) across 20 replicates evaluated at $\tau \in \{ 0.1, 0.2, 0.3, 0.4, 0.5, 0.6, 0.7, 0.8, 0.9\}$ for Simulations 1 through 3.}
	\label{fig:simu-quantile-main}
\end{figure}

Table \ref{tab:simulation-main} summarizes the results for PGQR with $\alpha \in \{ 1, 5 \}$ in \eqref{eq:pen} and all competing methods, averaged across 20 experiments. There was not much substantial difference between $\alpha = 1$ and $\alpha = 5$ for PGQR. Table \ref{tab:simulation-main} shows that PGQR had the lowest PMSE in all three simulations and attained coverage close to the nominal rate. In Simulations 2 and 3, the prediction intervals produced by PGQR had the highest coverage rate. In Simulation 1, FlexZBoost had 100\% coverage, but the average width of the prediction intervals for FlexZBoost was considerably larger than that of the other methods, suggesting that the intervals produced by FlexZBoost may be too conservative to be informative. 

Since PGQR approximates conditional quantile functions, we also compared PGQR with two other neural network approaches for nonparametric joint quantile regression. Specifically, we compared PGQR to the MCQRNN method of \cite{cannon2018non} and the $\ell_1$-penalization ($\ell_1$-p) method of \cite{moon2021learning}. These methods are available in the \textsf{R} packages \texttt{qrnn} and \texttt{l1pm} respectively.

To compare PGQR to $\ell_1$-p and MCQRNN, we considered the quantile accuracy criterion used by \cite{moon2021learning}.
Given $K$ prespecified quantile levels, let $Q_{\tau_k}(\cdot)$ be the conditional quantile function for $k=1, \ldots, K$, and let $\widehat{Q}_{\tau_k}$ denote the associated estimated quantile function.
Given $\bm{X}$, we obtained the quantile accuracy,	$\widehat{r}_k(\bm{X};\widehat{Q}_{\tau_k})= B^{-1}\sum_{b=1}^B\mathbb{I}(\bm{Y}_b^*\leq \widehat{Q}_{\tau_k}(\bm{X}))$, where $\bm{Y}^*_{1:B}$ were sampled from $P( Y \mid \bm{X})$ with $B = 2000$ and $K=9$ ($\tau_k \in \{0.1,0.2,0.3,0.4,0.5,0.6,0.7,0.8,0.9 \}$).
The relative frequency $\widehat{p}_k(\bm{X}, \widehat{Q}_{\tau_k})$ between $\tau_{k-1}$ and $\tau_k$ was then calculated as $\widehat{r}_k(\bm{X};\widehat{Q}_{\tau_{k}})-\widehat{r}_k(\bm{X};\widehat{Q}_{\tau_{k-1}})$.
If $\widehat{Q}_{\tau_k}$ is identical to true conditional quantile function, then $\widehat{p}_k(\bm{X}, \widehat{Q}_{\tau_k}) \overset{p}{\rightarrow} 1/(K+1) \text{ as }B \rightarrow \infty$.
Based on this fact, we used total variation distance (TV) and Hellinger distance (HD) to measure the distance between $\{ \widehat{p}_{k}(\bm{X}, \widehat{Q}_{\tau_k}) \}_{k=1,\ldots,K}$ and $\{1/(K+1),\ldots, 1/(K+1) \}$.
A more detailed description can be found in \cite{moon2021learning}.

In Figure \ref{fig:simu-quantile-main}, we plot the performance of PGQR, $\ell_1$-p, and MCQRNN for Simulations 1 through 3. Additional simulation results are provided in Appendix \ref{sec:sims}. Overall, PGQR had comparable performance to $\ell_1$-p, and both PGQR and $\ell_1$-p performed better than MCQRNN with lower average TV and HD. We reiterate that a major difference between these methods is that PGQR uses a deep \emph{generative} model to \emph{generate} the conditional quantiles from many random quantile levels $\tau$'s, whereas $\ell_1$-p and MCQRNN require the practitioner to specify $K$ target quantile levels $\tau_k, k = 1, \ldots, K$, at which to estimate the conditional quantiles.




\subsection{Real Data Analysis}

We examined the performance of PGQR on three real datasets from the UCI Machine Learning Repository, which we denote as: \texttt{machine}, \texttt{fish}, and \texttt{noise}.\footnote{Accessed from \url{https://archive.ics.uci.edu/ml/index.php}.}
We also applied PGQR to \texttt{YVRprecip} dataset analyzed by \cite{cannon2018non} and \cite{moon2021learning}. The \texttt{machine} dataset comes from a real experiment that collected the excitation current ($Y$) and four machine attributes ($\bm{X}$) for a set of synchronous motors \citep{kahraman2014metaheuristic}. The \texttt{fish} dataset contains the concentration of aquatic toxicity ($Y$) that can cause death in fathead minnows and six molecular descriptors ($\bm{X}$) described in \cite{cassotti2015similarity}. The \texttt{noise} dataset measures the scaled sound pressure in decibels ($Y$) at different frequencies, angles of attacks, wind speed, chord length, and suction side displacement thickness ($\bm{X}$) for a set of airfoils \citep{lopez2008neural}. 
Finally, the \texttt{YVRprecip} dataset collects daily precipitation totals (mm) at Vancouver International Airport from 1971 to 2000. The five covariates in \texttt{YVRprecip} are seasonal cycle, daily sea-level pressures, 700-hPa specific humidities, and 500-hPa geopotential heights \citep{cannon2018non, moon2021learning}. The \texttt{noise} and \texttt{YVRprecip} datasets display extreme skewness in the responses, rendering them especially challenging for conditional quantile and conditional density estimation.


\begin{table}[t]
	\centering
	\resizebox{0.8\columnwidth}{!}{%
		\begin{tabular}{ccc|cc|cc|cc|cc}
			\hline
			&&&\multicolumn{2}{c|}{PGQR}&\multicolumn{2}{c|}{GCDS}&\multicolumn{2}{c|}{deep-GCDS}&\multicolumn{2}{c}{WGCS}\\
			\hline
			Dataset&$n$&$p$&Cov&Width&Cov&Width&Cov&Width&Cov&Width\\
			\hline
			\hline
			machine&557 &4 &\textbf{0.96} &2.82&0.91 &1.92&0.87&1.52&0.69&1.66\\
			fish&908&6 &\textbf{0.96} &3.29&0.79 &2.38&0.54&1.08&0.80&2.36\\
			noise&1503&5 &\textbf{0.92} &7.58 &0.00&1.68&0.00&0.24&0.00&22.41 \\
            YVRprecip&10958&5 &\textbf{0.94} &12.6 &0.73&3.62&0.89&16.1&0.06&2.49 \\
			\hline
	\end{tabular}}
	\vspace{-0.1cm}
	\caption{\small Results from our real data analysis. Cov and width denote the coverage rate and average width respectively of the 95\% prediction intervals for the test observations.}
	\label{tab:real}
\end{table}

We examined the out-of-sample performance for PGQR (with fixed $\alpha = 1$), GCDS, deep-GCDS, and WGCS. In particular, 80\% of each dataset was randomly selected as training data, 10\% was used as validation data for tuning parameter selection, and the remaining 10\% was used as test data for model evaluation. To compare these deep generative methods, we considered the out-of-sample coverage rate (Cov) and the average width of the 95\% prediction intervals in the test data. 

The results from our real data analysis are summarized in Table \ref{tab:real}. On all four datasets, PGQR achieved higher coverage that was closer to the nominal rate than the competing methods. It seems as though the other generative approaches may have been impacted by the vanishing variability phenomenon, resulting in too narrow prediction intervals that did not cover as many test samples. In particular, GCDS, deep-GCDS, and WGCS all performed very poorly on the \texttt{noise} dataset, with an out-of-sample coverage rate of zero. On the other hand, with the help of the variability penalty \eqref{eq:pen}, PGQR did \textit{not} underestimate the variance and demonstrated an overwhelming advantage over these other methods in terms of predictive power. Moreover, the average widths of the PGQR prediction intervals were not overwhelmingly large so as to be uninformative.

Additional illustrations and data analyses are provided in Appendix \ref{sec:real-data}. In particular, Appendix \ref{sec:real-data} presents a clinically important application of PGQR for predicting muscular strength in older adults.



\section{Discussion} \label{sec:conclusion}

In this paper, we have made contributions to both the quantile regression and deep learning literature. Specifically, we proposed PGQR as a new deep generative approach to joint quantile estimation and conditional density estimation. Different from existing conditional sampling methods \citep{zhou2022deep, liu2021wasserstein}, PGQR employs a novel variability penalty to counteract the \textit{vanishing variability} phenomenon in deep generative networks. We introduced the PMNN family of neural networks to enforce the monotonicity of quantile functions. Finally, we provided a scalable implementation of PGQR which requires solving only a single optimization to select the regularization term in the variability penalty. Through analyses of real and simulated datasets, we demonstrated PGQR's ability to capture critical aspects of the conditional distribution such as multimodality, heteroscedasticity, and skewness. 

We anticipate that our penalty on vanishing variability in generative networks is broadly applicable for a wide number of loss functions besides the check function. In the future, we will extend the variability penalty to other deep generative models and other statistical problems besides quantile regression. In addition, we plan to pursue variable selection, so that our method can also identify the most relevant covariates when the number of features $p$ is large. Owing to its single-model training, PGQR is scalable for large $n$. However, further improvements are needed in order for PGQR to avoid the curse of dimensionality for large $p$.

\section*{Acknowledgments}
We are grateful to the two anonymous reviewers whose thoughtful feedback helped us greatly improve our paper. We also thank Dr. Jun Liu from Harvard University for his helpful comments. We also thank Dr. Jun Liu from Harvard University for his helpful comments. The last author was generously supported by NSF grant DMS-2015528.

\bibliographystyle{apalike}
    \bibliography{{PGQR-ref}}

\newpage 

\appendix

\section{Additional Illustrations and Real Data Analyses}\label{sec:real-data}

\subsection{Illustration: Takeuchi's Example} \label{sec:illustration}

A popular illustrative example in the literature for nonparametric quantile estimation was given by \cite{takeuchi2006nonparametric} (henceforth known as Takeuchi's example), where
\begin{align*}\label{eq:tajeuchi}
	Y_i = \sin(\pi X_i)/(\pi X_i) + \epsilon_i, \quad i = 1,\ldots, n,
\end{align*}
with $x_i \sim \text{Uniform}(-1, 1)$ and $\epsilon_i \sim \mathcal{N} (0, 0.1\exp(1-X_i))$.
By \cite{takeuchi2006nonparametric}, the true $\tau$th conditional function, $\tau \in (0,1)$, is $f_{\tau}(X)=\sin(\pi X)/(\pi X)+0.1\exp(1-X)\Phi^{-1}(\tau)$, where $\Phi^{-1}(\cdot)$ denotes the inverse cumulative distribution function (cdf) of a standard Gaussian distribution. 
Moreover, the conditional distribution of $Y$ given $X$ is $Y \mid X \sim \mathcal{N}(\sin(\pi X)/(\pi X), 0.1\exp(1-X))$. 

We made an artificial dataset of size $n=2000$ for Takeuchi's example and applied our proposed PGQR method to it.
\begin{figure}[t!]
	\centering
	\includegraphics[width=1.0\textwidth]{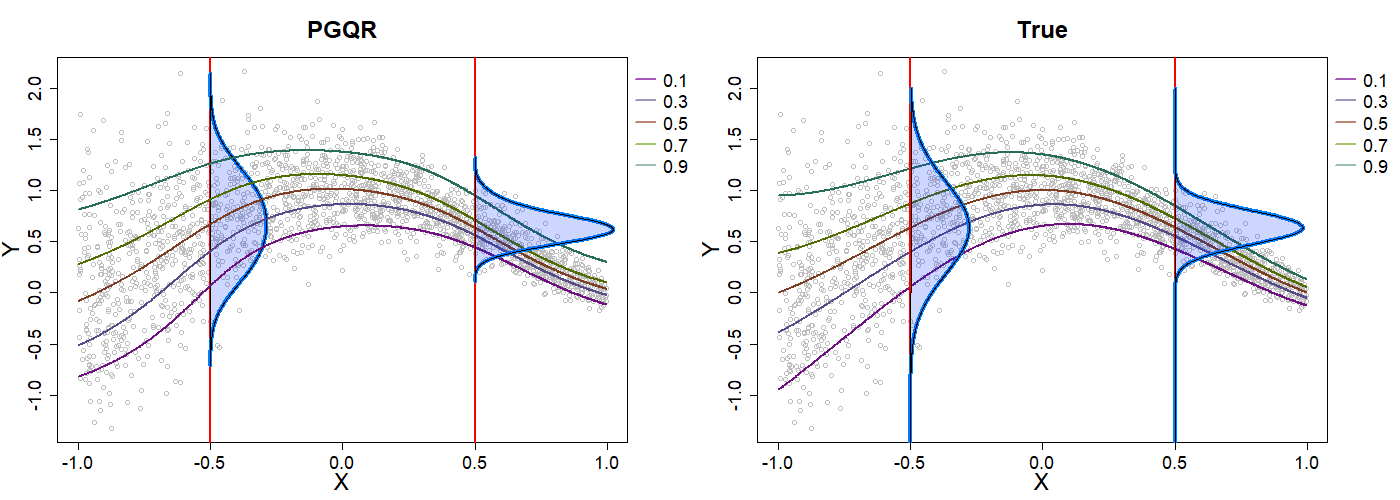}
	\caption{\small Using PGQR to model the conditional densities $p(Y\mid X=0.5)$ and $p(Y\mid X=-0.5)$. The conditional quantile functions at levels \{0.1, 0.3, 0.5, 0.7, 0.9\} are also displayed.}
    \vspace{-0.2cm}
	\label{fig:take}
\end{figure}
Specifically, we aimed to estimate the conditional densities $P(Y\mid X=-0.5)$ and $P(Y\mid X=0.5)$. The left panel of Figure \ref{fig:take} depicts the conditional density estimates for PGQR, which are almost identical to the true conditional densities (right panel of Figure \ref{fig:take}).
In addition, we estimated the conditional quantile function $f_{\tau}(x)$ at quantile levels $\tau=\{0.1, 0.3, 0.5, 0.9 \}$.
From Figure \ref{fig:take}, we observe that by comparing results of PGQR to the true quantile functions, the mid-range quantile levels $\{0.3, 0.5, 0.7 \}$ are estimated quite well. However, for the quantile levels $\{0.1, 0.9 \}$, PGQR exhibits a slight departure from truth, which is more obvious when $X$ is close to the boundary of the covariate domain (-1 and 1). Quantile levels very close to zero or one are inherently more challenging to estimate because there is less data in these regions.

\subsection{Clinical Application: Discovering Hidden Subpopulations} \label{sec:Motivation}

PGQR aims to simultaneously generate samples from multiple conditional quantiles $Q_{Y \mid \bm{X}} (\tau)$ of $p(Y \mid \bm{X})$ at different quantile levels $\tau \in (0,1)$. An automatic byproduct of joint \textit{nonparametric} quantile regression (as opposed to linear quantile regression) is that if the conditional quantiles $Q_{Y \mid \bm{X}} (\tau)$ are estimated well for a large number of quantiles, then we can also infer the \textit{entire} conditional distribution for $Y$ given $\bm{X}$. 

\begin{figure}[t]
	\centering
	\includegraphics[width=.9\textwidth]{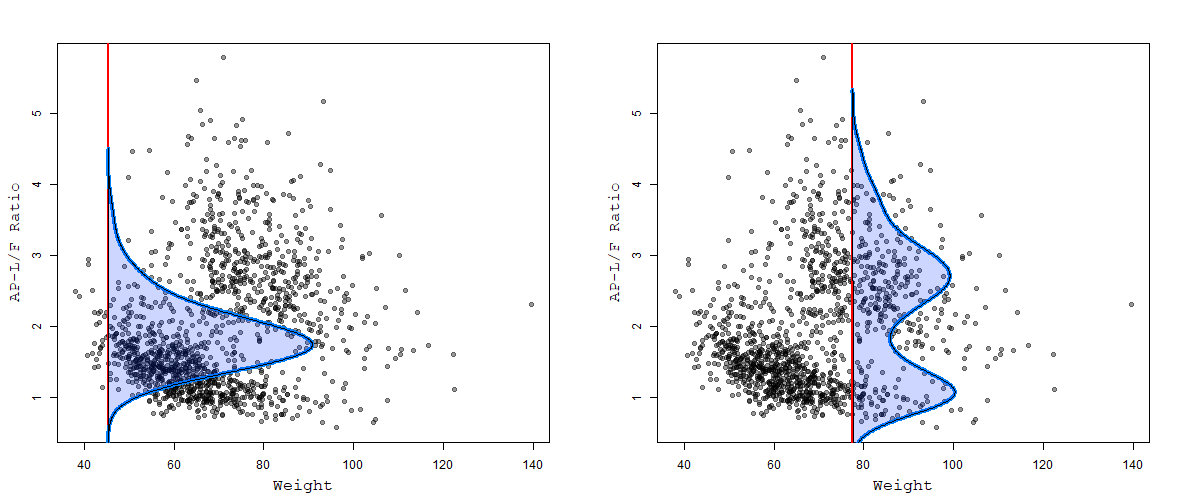}
	\caption{\small Using PGQR to model the conditional density of AP-L/F ratio given weight in older adults. Left panel: Weight = 45.4 kg, Right panel: Weight = 77.5 kg.}
	\label{fig:bimodal}
\end{figure}

To demonstrate the clinical utility of our method, we apply our proposed PGQR method to a real dataset on body composition and strength in older adults \citep{RoyChoudhury2020}. The data was collected over a period of 12 years for 1466 subjects as part of the Rancho Bernardo Study (RBS), a longitudinal observational cohort study. We are interested in modeling the appendicular lean/fat (AP-L/F) ratio, i.e.
\begin{align*}
	\text{AP-L/F Ratio} = \frac{\text{Weight on legs and arms}}{\text{Fat weight}},
\end{align*}
as a function of weight (kg). Accurately predicting the AP-L/F ratio is of practical clinical interest, since the AP-L/F ratio provides information about limb tissue quality and is used to diagnose sarcopenia (age-related, involuntary loss of skeletal muscle mass and strength) in adults over the age of 30 \citep{Evans2010AJCN, Scafoglieri2017}. 

Figure \ref{fig:bimodal} plots the approximated conditional density of AP-L/F ratio given weight of 45.4 kg (left panel) and 77.5 kg (right panel) under the PGQR model. We see evidence of data heterogeneity (actually, depending on an unobserved factor of gender), as the estimated conditional density is unimodal when the weight of older adults is 45.4 kg but \textit{bimodal} when the weight of older adults is 77.5 kg. In short, our method \textit{discovers} the presence of two heterogeneous subpopulations of adults that weigh around 78 kg. In contrast, mean regression (e.g. simple linear regression or nonparametric mean regression) of AP-L/F ratio given weight might obscure the presence of two modes and miss the fact that weight affects AP-L/F ratio differently for these two clusters of adults.

\section{More Simulation Results} \label{sec:sims}

\subsection{Additional Simulation Studies}\label{sec:sim-more}

In addition to the three simulations described in Section \ref{subsec:construction}, we also conducted simulation studies under the following scenarios: 
\begin{itemize}
	\item \noindent{\bf Simulation 4: Nonlinear function with an interaction term and one irrelevant covariate.} $Y_i = 0.5 \log(10-X_{i1}^2)+0.75 \exp(X_{i2} X_3/5)-0.25 \vert X_{i4}/2 \vert + \epsilon_i$, where $\epsilon_i \sim \mathcal{N}(0,1)$. Note that there is a (nonlinear) interaction between $X_2$ and $X_3$, while $X_5$ is irrelevant.
	\item \noindent{\bf Simulation 5: Very small conditional variance.} $Y_i = \beta X_i + \epsilon_i, i = 1, \ldots, n$, where $\beta = 1$ and $\epsilon_i \overset{iid}{\sim} \mathcal{N}(0, 0.01)$. 
    \item \noindent{\bf Simulation 6: Error term dependent on norm of predictor $\bm{X}$.} $Y_i = \bm{X}_i^{\top} \boldsymbol{\beta} + \epsilon_i$, where $\boldsymbol{\beta} \in \mathbb{R}^5$ is equispaced between $[-2, 2]$, $\bm{X}_i \sim \text{Uniform}[-1,1]^5$ and $\epsilon_i \sim \mathcal{N}(0, \exp(0.5\vert X_i\Vert_1))$. 
\end{itemize}
The results from these three simulations averaged across 20 replicates are shown in Table \ref{tab:supp-simulation}. 

\begin{table}[t]
	\centering
	\resizebox{1.0\columnwidth}{!}{
		\begin{tabular}{c|ccc|ccc|ccc}
			\hline
			&\multicolumn{3}{c|}{Simulation 4}&\multicolumn{3}{c}{Simulation 5}&\multicolumn{3}{c}{Simulation 6}\\
			\hline
			Method&$\mathbb{E}(Y \mid \bm{X})$&$\text{sd}(Y \mid \bm{X})$& Cov (Width)&$\mathbb{E}(Y \mid \bm{X})$&$\text{sd}(Y \mid \bm{X})$&Cov (Width)&$\mathbb{E}(Y \mid \bm{X})$&$\text{sd}(Y \mid \bm{X})$&Cov (Width)\\
			\hline
			\hline
			PGQR ($\alpha=1$)
			&0.15&0.07&0.95 (4.33)&0.004&\textbf{0.0001}&0.93 (0.42)&7.20&57.75&0.79(17.72)
			\vspace{0.05cm}\\
			PGQR ($\alpha=5$)
			&\textbf{0.14}&0.08&\textbf{0.96 (4.50)}&0.005&0.03&\textbf{0.99 (1.01)}&7.20&56.74&0.79(18.80)
			\vspace{0.05cm}\\
			GCDS
			&0.23&0.05&0.87 (3.39)&\textbf{0.002}&0.0011&0.73 (0.27)&7.44&64.63&\textbf{0.80}(18.41)
			\vspace{0.05cm}\\
			deep-GCDS
			&0.43&0.20&0.76 (2.88)&0.004&0.0022&\textbf{0.99 (0.56)}&11.9&85.60&0.61(13.46)
			\vspace{0.05cm}\\
			WGCS
			&0.92&0.15&0.79 (4.41)&0.640&0.2372&0.70 (1.15)&9.58&93.52&0.51(11.54)
			\vspace{0.05cm}\\
			\hline
			\hline
			FlexCoDE-NNR
			&0.23&0.003&0.92 (3.82)&0.069&0.0168&0.89 (0.27)&\textbf{1.57}&\textbf{35.46}&0.37(8.38)
			\vspace{0.05cm}\\
			FlexCoDE-SAM
			&0.23&\textbf{0.002}&0.93 (3.83)&0.043&0.0130&0.93 (4.04)&1.76&50.72&0.09(5.56)
			\vspace{0.05cm}\\
			FlexZBoost
			&0.45&0.06&0.82 (3.50)&1.244&0.2824&0.81 (4.13)&19.1&40.36&0.58(15.5)
			\vspace{0.05cm}\\
			RFCDE
			&0.24&0.003&0.94 (3.97)&0.067&0.0292&0.92 (3.97)&3.12&67.33&0.28(6.61)
			\\
			\hline
	\end{tabular}}
	\vspace{-0.3cm}
	\caption{\small  Table reporting the PMSE for the conditional expectation and standard deviation, as well as the coverage rate (Cov) and average width of the 95\% prediction intervals, for Simulations 4 through 6. Results were averaged across 20 replicates.}
	\label{tab:supp-simulation}
	
\end{table}


One may be concerned whether the regularized PGQR overestimates the conditional variance when the true conditional density has a very \textit{small} variance. To illustrate the flexibility of PGQR, Figure \ref{fig:simu-smallvar} plots the estimated conditional densities for three test observations from one replication of Simulation 5. Recall that in Simulation 5, the true conditional variance is very small ($\sigma^2 = 0.01$). With the optimal $\lambda^{\star}$ selected using the method introduced in Section \ref{sec:select}, Figure \ref{fig:simu-smallvar} shows that the estimated PGQR conditional density \textit{still} manages to capture the Gaussian shape while matching the true variance of 0.01. 
If the true conditional variance is very small (as in Simulation 5), then PGQR selects a tiny $\lambda^{\star} \approx 0$. In this scenario, PGQR only applies a small amount of variability penalization and thus does not overestimate the variance. 

In Simulation 6, we investigated the especially challenging case when the error term $\epsilon$ is dependent on $\ell_1$ norm of predictor $\bm{X}$.
This simulation setting is inspired by simulation M2 in \cite{moon2021learning} and is a variant of an example from Appendix A of \cite{takeuchi2006nonparametric}.
In Figure \ref{fig:simu-norm}, we see that the variance of true conditional distribution varies with $\Vert \bm{X} \Vert_1$.
We observe that PGQR was able to capture of true conditional distribution in some cases where $\Vert \bm{X} \Vert_1$ is not large (e.g. $\Vert \bm{X} \Vert_1 < 4$).
However, PGQR ($\alpha=1$) is incapable of estimating very large variance in the third graph (upper panel) when $\Vert \bm{X} \Vert_1=6.4$. In this case, the true conditional density is very flat and thus difficult for all of the deep generative methods to estimate well. It is worth noting that the other state-of-the art methods, GCDS and WGCS, struggled even more than PGQR in this heavy heteroscedasticity scenario.

\begin{figure}[t!]
	\centering
	\includegraphics[width=\textwidth]{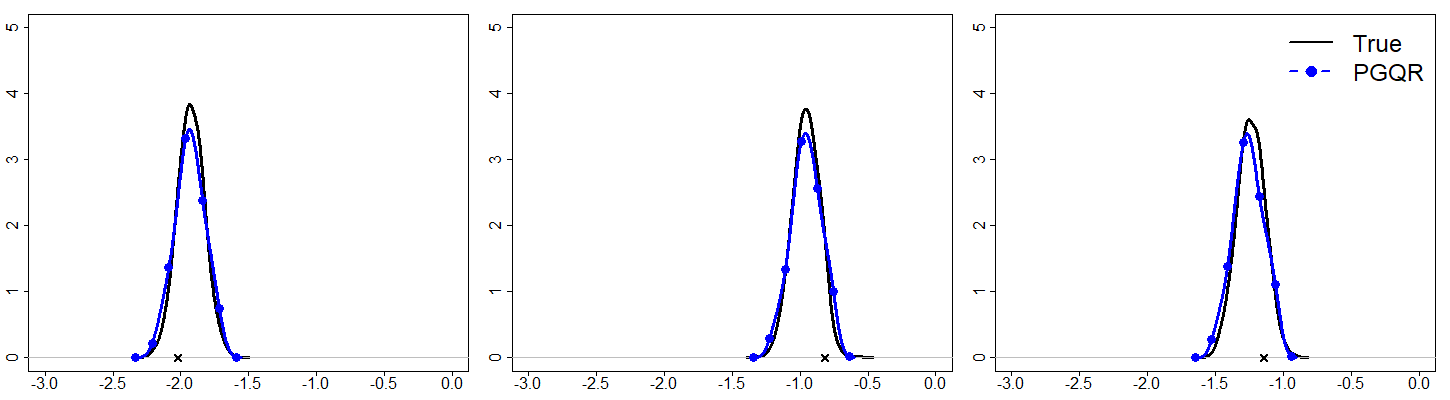}
	\caption{\small Plots of the estimated PGQR $(\alpha=1)$ conditional densities $p(Y \mid \bm{X}_{\text{test}})$ for three different test observations from one replication of Simulation 5. The optimal $\lambda^{\star}$ is chosen by our tuning parameter selection method in Section \ref{sec:select}.}
	\label{fig:simu-smallvar}
\end{figure}

\begin{figure}[H]
	\centering
	\includegraphics[width=\textwidth]{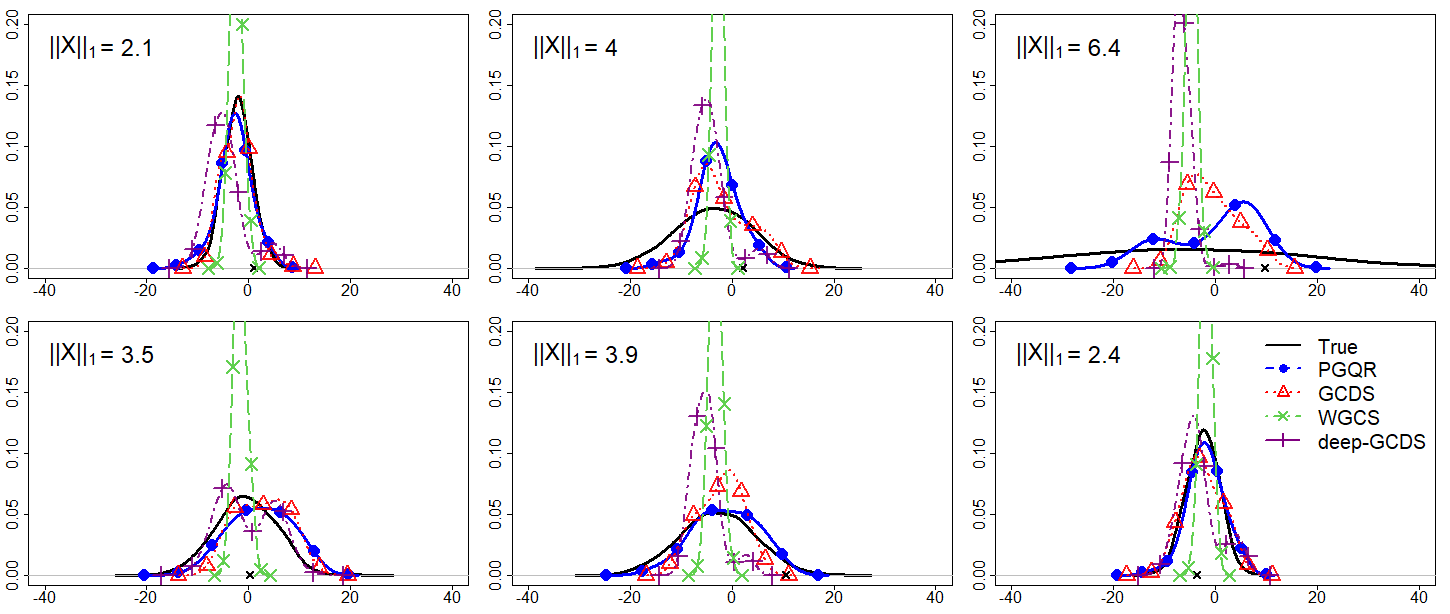}
	\caption{\small Plots of the estimated PGQR ($\alpha=1$) conditional densities $p(Y \mid \bm{X}_{\text{test}})$ for six different test observations from one replication of Simulation 6.}
	\label{fig:simu-norm}
\end{figure}



\subsection{Additional Figures} \label{sec:sim-morefigs}

Here, we provide additional figures from one replication each of Simulations 1, 2, and 4 (with $\alpha=1$ in PGQR). Figures \ref{fig:simu-quantile-sim1}-\ref{fig:simu-quantile-sim4} illustrate that PGQR (blue solid line with filled circles)  is better able to estimate the true conditional densities (solid black line) than GCDS, WGCS, and deep-GCDS (dashed lines). In particular, PGQR does a better job of capturing critical aspects of the true conditional distributions such as multimodality, heteroscedasticity, and skewness. 

\newpage

\begin{itemize}[leftmargin=.2in]
	\setlength\itemsep{0.2em}
	\item \noindent{\bf Simulation 1: Multimodal and heteroscedastic.}
	\begin{figure}[H]
		\centering
		\includegraphics[width=.9\textwidth]{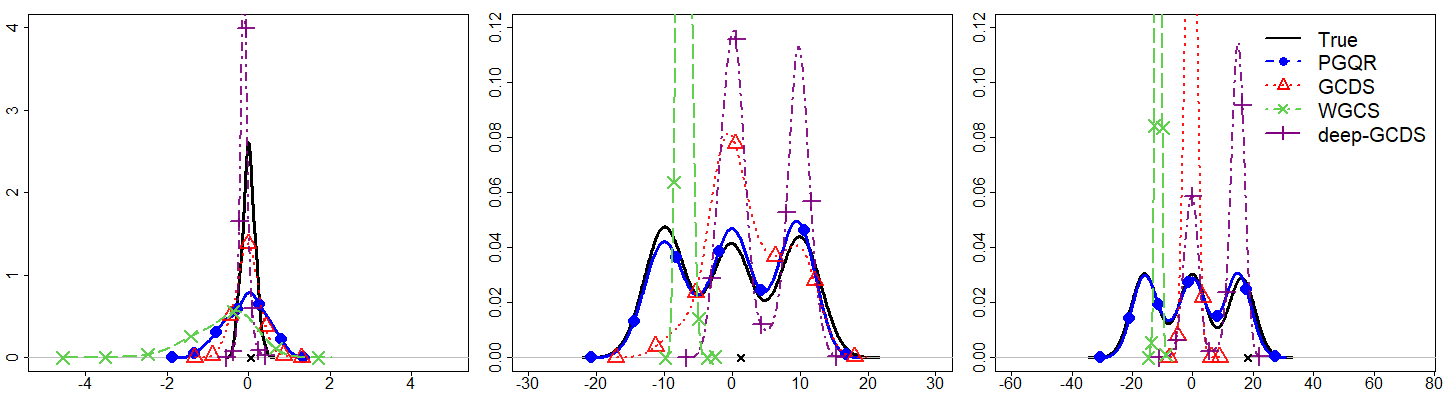}
		\caption{\small Plots of the estimated conditional densities $p(Y \mid \bm{X}_{\text{test}})$ for three different test observations from one replication of Simulation 1.}
		\label{fig:simu-quantile-sim1}
	\end{figure}

	\item \noindent{\bf Simulation 2: Mixture of left-skewed and right-skewed.}
	\begin{figure}[H]
		\centering
		\includegraphics[width=.9\textwidth]{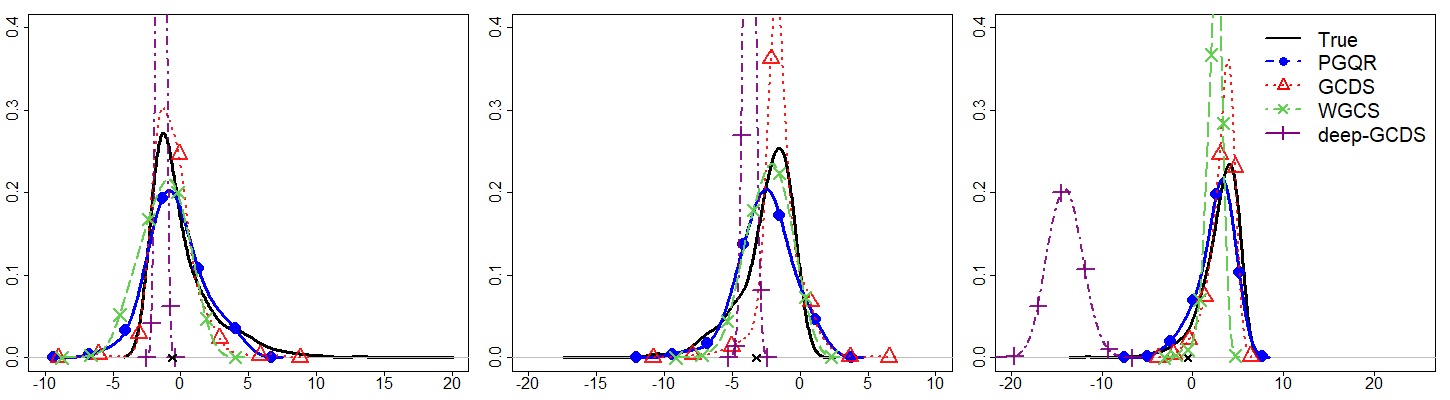}
		\caption{\small Plots of the estimated conditional densities $p(Y \mid \bm{X}_{\text{test}})$ for three different test observations from one replication of Simulation 2.}
		\label{fig:simu-quantile-sim2}
	\end{figure}    
	\item \noindent{\bf Simulation 4: Nonlinear function with an interaction term and one irrelevant covariate.}
	\begin{figure}[H]
		\centering
		\includegraphics[width=.9\textwidth]{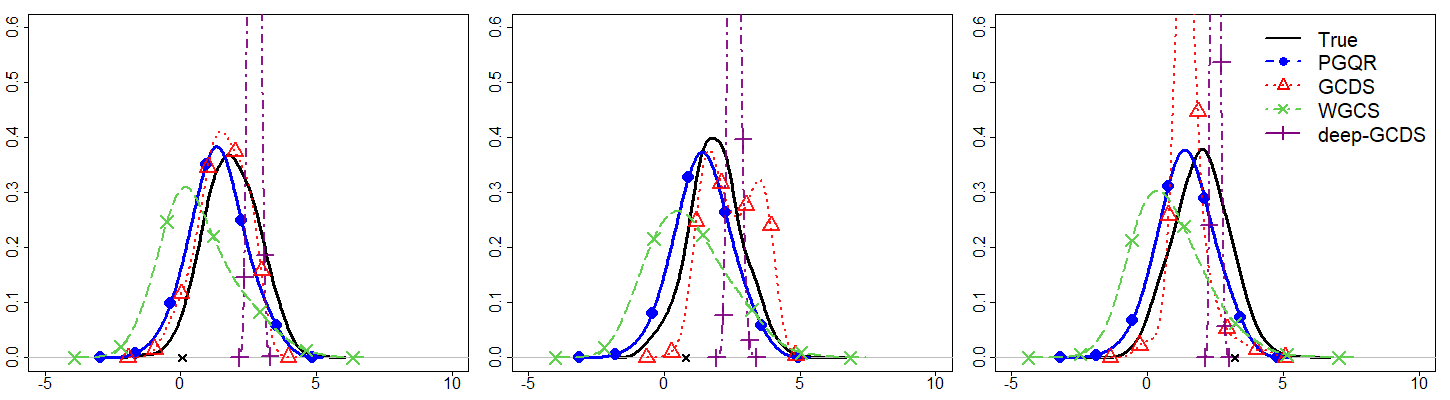}
		\caption{\small Plots of the estimated conditional densities $p(Y \mid \bm{X}_{\text{test}})$ for three different test observations from one replication of Simulation 4.}
		\label{fig:simu-quantile-sim4}
	\end{figure}
\end{itemize}

\begin{figure}[H]
	\centering
	\includegraphics[width=.98\textwidth]{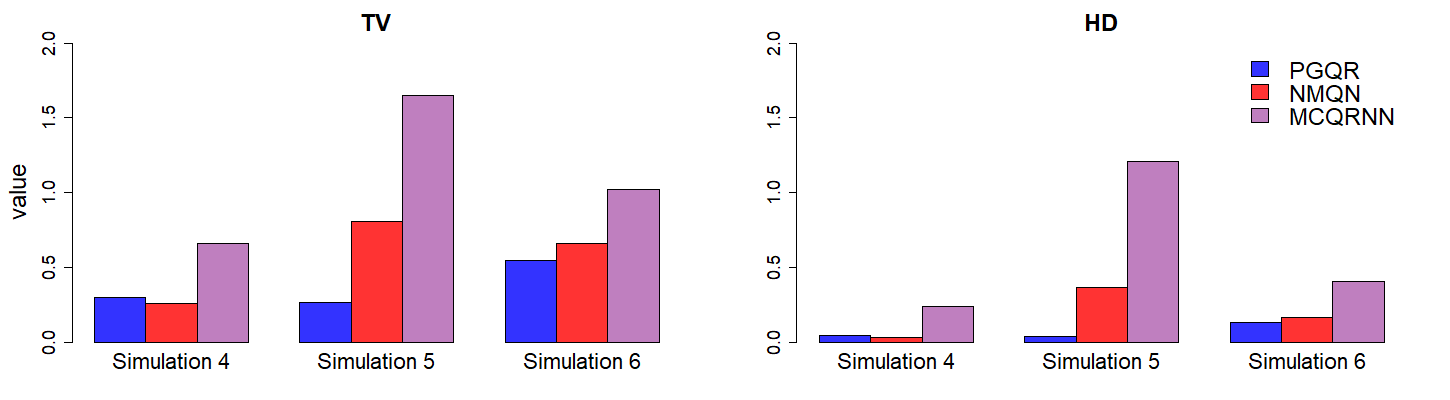}
	\caption{\small Barplots of the average total variation distance (TV) and Hellinger distance (HD) across 20 replicates evaluated at $\tau \in \{ 0.1, 0.2, 0.3, 0.4, 0.5, 0.6, 0.7, 0.8, 0.9\}$ for Simulations 4 through 6.}
	\label{fig:simu-quantile-ap}
\end{figure}

Figure \ref{fig:simu-quantile-ap} depicts the quantile estimation performance of PGQR compared to $\ell_1$-p and MCQRNN in Simulations 4 through 6. In Simulation 4, $\ell_1$-p performed slightly better than PGQR, but the performance of the two methods was very comparable. However, in Simulations 5 and 6, PGQR outperformed $\ell_1$-p, with lower average total variation distance (TV) and Hellinger distance (HD). Both PGQR and $\ell_1$-p outperformed MCQRNN in Simulations 4 through 6.

\section{Proofs of Propositions} \label{sec:proofs}

{\bf Proof of Proposition \ref{prop:PGQRexistence}.}
Suppose that for some $\epsilon>0$, there exists a set $\mathcal{C}\subset (0,1)$ with $P_\tau(\mathcal{C})>\epsilon$ such that for some $i^*\in\{1,\dots,n\}$, $\widehat g_\tau({\bm X}_{i^*}) \neq \widehat G({\bm X}_{i^*},\tau)$ for all $\tau\in \mathcal{C}$. Then we can construct another optimal generator $\widetilde G$ such that 
\begin{eqnarray*}
	&&\frac{1}{n} \sum_{i=1}^n \mathbb{E}_{\tau} \left[ \rho_{\tau} \big(y_i - \widehat G(\bm{X}_i, \tau) \big)\right] + \mathbb{E}_{\widetilde\tau, \widetilde\tau^\prime} \left\{\text{pen}_{\lambda, \alpha}\left( \widehat G(\bm{X}_i,\widetilde\tau), \widehat G(\bm{X}_i, \widetilde\tau^\prime) \right) \right\}\\
	&\geq& \frac{1}{n} \sum_{i=1}^n \mathbb{E}_{\tau} \left[ \rho_{\tau} \big(y_i - \widetilde G(\bm{X}_i, \tau) \big)\right] + \mathbb{E}_{\widetilde\tau, \widetilde\tau^\prime} \left\{\text{pen}_{\lambda, \alpha}\left( \widetilde G(\bm{X}_i,\widetilde\tau), \widetilde G(\bm{X}_i, \widetilde\tau^\prime) \right) \right\},
\end{eqnarray*}
where $\widetilde G(\bm{X}_{i^*},\tau) = \widehat G(\bm{X}_{i^*},\tau)$ for $\tau\not\in \mathcal{C}$ and $\widetilde G(\bm{X}_{i^*},\tau) = \widehat g_\tau(\bm{X}_{i^*})$ for $\tau\in \mathcal{C}$. This is a contradiction due to the fact that $\widehat g_\tau$ is the minimizer of $\frac{1}{n} \sum_{i=1}^n \rho_{\tau} (y_i - \widehat G(\bm{X}_i, \tau) ) + \mathbb{E}_{\widetilde\tau, \widetilde\tau^\prime} \{ \text{pen}_{\lambda, \alpha} ( \widehat G(\bm{X}_i,\widetilde\tau), \widehat G(\bm{X}_i, \widetilde\tau^\prime) ) \}$. \qed \\

\noindent {\bf Proof of Proposition \ref{prop:nomemorization}.}  Suppose that for all $\lambda \geq 0$, all $\tau \in (0,1)$, and all $i\in \{1,\dots,n\}$, 
\begin{eqnarray*}
	\widehat G({\bm X}_i,\tau)=Y_i.
\end{eqnarray*}
Then the penalty part in the PGQR loss \eqref{eq:PGQR} attains $\mathbb{E}_{\tau, \tau^\prime} \{ \text{pen}_{\lambda, \alpha} ( \widehat G(\bm{X}_i,\tau), \widehat G(\bm{X}_i, \tau^\prime) ) \}= \lambda \log(\alpha)$, while the first term in \eqref{eq:PGQR} satisfies $\mathbb{E}_{\tau} [ \rho_{\tau} \big(y_i - \widehat G(\bm{X}_i, \tau) \big)]=0$. As a result, the total loss is $\lambda \log(\alpha)$. 

Since the case of $\widehat G({\bm X}_i,\tau)=Y_i$ is included in cases of $\textrm{Var}_\tau\{\widehat G({\bm X}_i,\tau)\}=0$, we focus on the variance. When there exists some $i \in \{1,\dots,n\}$ and some $\tau\in (0,1)$ such that
$$
\textrm{Var}_\tau\left\{\widehat G({\bm X}_i,\tau)\right\}>0,
$$
the resulting total loss can be made less than $\lambda\log(\alpha)$ by choosing an appropriate $\lambda>0$. This contradicts the fact that $\widehat G$ is the minimizer as in \eqref{eq:PGQR}. \qed \\

\noindent {\bf Proof of Proposition \ref{prop:lambda}.} 
It is trivial that $F_Q(W) := P(Q\leq W\mid W) \sim \text{Uniform}(0,1)$ when $Q\overset{d}=W$. Without loss of generality, we assume that $F_Q$ is invertible. If $F_Q$ is not invertible, then we replace $F_Q^{-1}(W)$ below with $F_Q^{-}(W) := \inf \{ x: F_Q(x) \geq W \}$, and the result still holds. We shall show that $F_Q(W)\sim \text{Uniform}(0,1)$ implies that the two distributions for $Q$ and $W$ are identical. Suppose that $F_Q(W)$ follows a standard uniform distribution. Then,
\begin{equation*}
	x = P(F_Q(W)<x) = P(W < F_Q^{-1}(x)) = F_Q(F_Q^{-1}(x)),
\end{equation*}
which implies that $F_W(F_W^{-1}(x)) = x$. Thus, it follows that the distributions of $Q$ and $W$ are identical. \qed 

\section{Analyses of Model Complexity and Choice of $\alpha$}
\label{sec:furtheranalysis}

\subsection{Model Complexity Analysis}\label{sec:model-complex}

The estimated conditional quantile function $\widehat{G}(\bm{X}, \tau)$ depends on estimation of two sub-networks $g_c(\tau)$ and $g_{uc}(\bm{X})$.
To avoid vanishing variability \eqref{eq:overfit}, we proposed a novel variability penalty \eqref{eq:pen} which essentially controls $\lVert \partial G(\bm{X}, \tau) / \partial \tau \rVert$.
However, the estimated $\widehat{G}(\bm{X}, \tau)$ also depends on model complexity $\lVert \partial G(\bm{X}, \tau) / \partial \bm{X} \rVert$. This indicates that tuning the model complexity of $g_{uc}(\bm{X})$ might be another potential way to solve the vanishing variability problem. 

In Section \ref{sec:pen}, we conducted a simple simulation study under the model, $Y_i = \bm{X}_i^\top \boldsymbol{\beta} + \epsilon_i, i = 1, \ldots, n$,
where $\bm{X}_i \overset{iid}{\sim} \mathcal{N}(\boldsymbol{0},\bm{I}_{20})$, $\epsilon_i \overset{iid}{\sim} \mathcal{N}(0,1)$, the coefficients in $\boldsymbol{\beta}$ are equispaced over $[-2, 2]$, and the sample size is $n=2000$. Figure \ref{fig:overfit} illustrated that GQR (without any variability penalty) can encounter severe vanishing variability -- essentially collapsing to a point mass -- even when dealing with a simple linear model. We also showed in Section \ref{sec:pen} that PGQR (i.e. GQR with a variability penalty function) avoids the vanishing variability by controlling $\lVert \partial G( \bm{X}, \tau) / \partial \tau \rVert$.
In this example, both GQR and PGQR were constructed by a feedforward neural network (FNN) with three hidden layers and 1000 neurons per hidden layer. Here, GQR is heavily overparameterized, i.e. the number of learnable parameters in the FNN is much larger than the number of training samples. 
As discussed in Section \ref{sec:pen}, the main motivation for overparameterization is that it improves the generalization and robustness of the model \citep{allen2019learning, zhang2021understanding, Soltanolkotabi2019, MontanariZhong2022}. But in the case of nonparameteric quantile estimation, overparameterization can also lead to more severe vanishing variability. This is because the GQR loss \eqref{eq:GQR} actually achieves the minimum value of zero when $\widehat{G}(\bm{X}_i, \tau) = Y_i, i = 1, \ldots, n$, and a very complex model is likely to perfectly interpolate the observed responses.

Therefore, instead of using a variability penalty on an overparameterized model, it is also very natural to consider controlling the model complexity $\lVert \partial G(\bm{X}, \tau) / \partial \bm{X} \rVert$ by simply choosing a simpler FNN structure $g_{uc}(\bm{X})$. A simpler model would not be overparameterized, and therefore, it is a promising alternative way to avoid vanishing variability. To investigate whether tuning model complexity indeed helps to avoid this phenomenon, we considered two different (simpler) FNN settings for the simple linear example that we presented in Section \ref{sec:pen}. We applied these simple FNN structures to (non-penalized) GQR so we could see the impact of $\lVert \partial G(\bm{X}, \tau) / \partial \bm{X} \rVert$ on vanishing variability.

To be more specific, we denote the model $\text{GQR}_1$ as (non-penalized) GQR fit with a simple FNN with two hidden layers, each having 50 hidden neurons.
This setting is similar to the network structure in \cite{zhou2022deep}.
The model $\text{GQR}_2$ is (non-penalized) GQR fit with an even simpler FNN architecture with only one hidden layer and five hidden neurons.
This setting is similar to that considered by \cite{cannon2018non} and \cite{moon2021learning}. We compared $\text{GQR}_1$ and $\text{GQR}_2$ to PGQR on the same out-of-sample test data, where the optimal penalty parameter $\lambda^{\star}$ in PGQR was chosen according to Algorithm \ref{alg:PGQR}. These results are displayed in Figure \ref{fig:overfit2} ($\text{GQR}_1$ vs. PGQR) and Figure \ref{fig:overfit1} ($\text{GQR}_2$ vs. PGQR). 

\begin{figure}[t!]
\centering
\includegraphics[width=.95\textwidth]{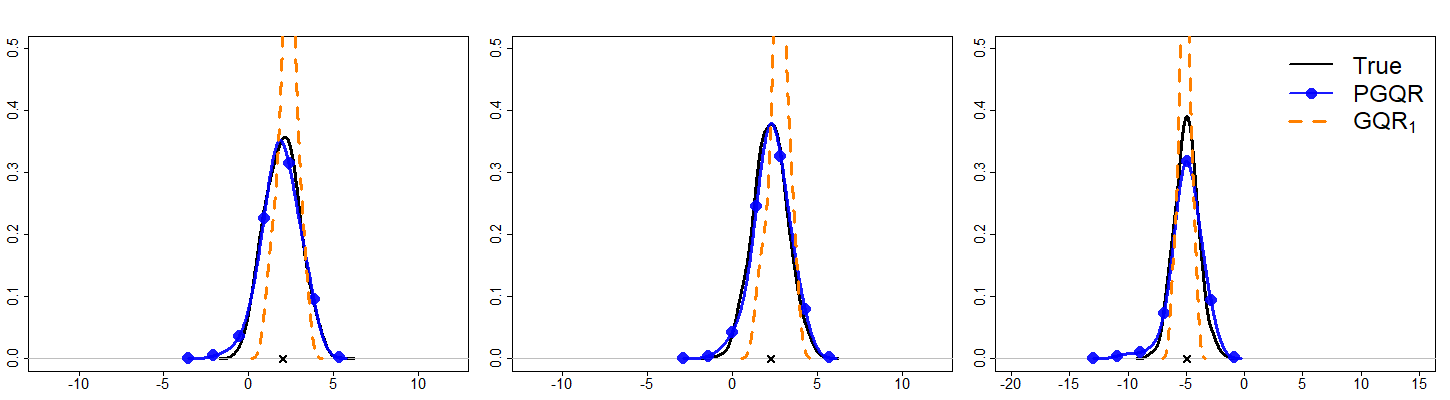}
\caption{\small Plots of the estimated conditional densities $p(Y \mid \bm{X}_{\text{test}})$ for three different test observations under $\text{GQR}_1$ and PGQR.}
\label{fig:overfit2}
\end{figure}

\begin{figure}[t!]
\centering
\includegraphics[width=.95\textwidth]{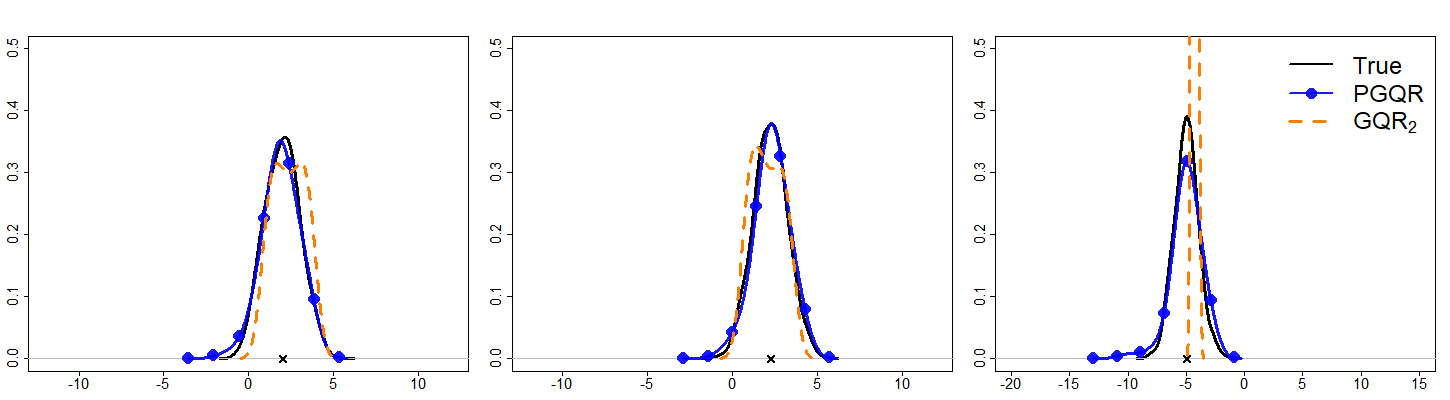}
\caption{\small Plots of the estimated conditional densities $p(Y \mid \bm{X}_{\text{test}})$ for three different test observations under $\text{GQR}_2$ and PGQR.}
\label{fig:overfit1}
\end{figure}

It is obvious from Figure \ref{fig:overfit2} that the simpler model $\text{GQR}_1$ (non-penalized GQR with two hidden layers and 50 nodes per hidden layer) mitigates vanishing variability a little bit, but it still suffers from this problem by underestimating the true variance. Looking at the results for $\text{GQR}_2$ (non-penalized GQR with one hidden layer and 5 hidden nodes) in Figure \ref{fig:overfit1}, the first two plots indicate that an even simpler FNN structure helps to avoid vanishing variability phenomenon for these particular test points. However, the third plot in Figure \ref{fig:overfit1} shows that $\text{GQR}_2$ can \emph{still} encounter the vanishing variability problem for other test observations, with significant variance underestimation.

From this simple example, we can see that even though we applied an extremely shallow FNN, vanishing variability can still occur. However, reducing the FNN complexity does appear to make vanishing variability less pronounced. Although a very simple model might not result in \emph{exact} interpolation of the training labels, it nevertheless does \emph{not} entirely fix variance underestimation. In short, controlling $\lVert \partial G(\bm{X}, \tau) / \partial \bm{X} \rVert$ by tuning the FNN complexity for $g_{uc}(\bm{X})$ \emph{fails} to completely avoid vanishing variability. This simple example demonstrates the distinct need to use a variability penalty to control $\lVert \partial G(\bm{X}, \tau) / \partial \tau \rVert$ (rather than just $\lVert \partial G(\bm{X}, \tau) / \partial \bm{X} \rVert$).

Figures \ref{fig:overfit2} and \ref{fig:overfit1} also show that the overparameterized PGQR model does a better job recovering the true conditional density $p(Y \mid \bm{X})$ -- and therefore also estimates the true conditional quantile functions better -- than the (non-penalized) GQR models with simpler FNN structures $g_{uc}(\bm{X})$. This may be because a simple neural network inevitably has less expressive power than a more complex one. By using (overparameterized) PGQR with a variability penalty, we are not only free of vanishing variability, but we \textit{also} fully realize the well-known benefits of overparameterization \citep{allen2019learning, zhang2021understanding, Soltanolkotabi2019, MontanariZhong2022}.

\subsection{Sensitivity Analysis of PGQR to the Choice of $\alpha$ }\label{sec:Robustana}
As mentioned in Section \ref{sec:pen}, we choose to fix $\alpha > 0$ in the variability penalty \eqref{eq:pen}. The main purpose of $\alpha$ is to ensure that the logarithmic term in the penalty is always well-defined. In this section, we conduct a sensitivity analysis to the choice of $\alpha$. To do this, we generated data using the same settings from Simulations 1 through 6. 
We then fit PGQR with eight different choices for $\alpha \in \{ 0.5, 1, 5, 10, 20, 30, 40, 50 \}$ and evaluated the performance of these eight PGQR models on out-of-sample test data.
Table \ref{tab:robustana} shows the results from our sensitivity analysis averaged across 20 replicates. 

We found that in Simulations 1 through 4 and Simulation 6, PGQR was not particularly sensitive to the choice of $\alpha$. PGQR was somewhat more sensitive to the choice of $\alpha$ in Simulation 5 (i.e. when the true conditional variance is very small), with larger values of $\alpha$ leading to higher PMSE for the conditional expectation and conditional standard deviation. In practice, we recommend fixing $\alpha = 1$ as the default $\alpha$ for PGQR to perform well. This choice of $\alpha =1$ leads to good empirical performance across many different scenarios.

\begin{table}[H]
\centering
\begin{subtable}[t]{\linewidth}
\centering
		\resizebox{1.0\textwidth}{!}{
			\begin{tabular}{c|ccc|ccc|ccc|}
				\hline
				&\multicolumn{3}{c|}{Simulation 1}&\multicolumn{3}{c|}{Simulation 2}&\multicolumn{3}{c|}{Simulation 3}\\
				\hline
				$\alpha$&$\mathbb{E}(Y \mid \bm{X})$&$\text{sd}(Y \mid \bm{X})$&Cov (Width)&$\mathbb{E}(Y \mid \bm{X})$&$\text{sd}(Y \mid \bm{X})$&Cov (Width)&$\mathbb{E}(Y \mid \bm{X})$&$\text{sd}(Y \mid \bm{X})$&Cov (Width)\\
				\hline
				\hline
				0.5
                &0.39&0.41&0.95 (23.60)&0.43&0.18&0.91 (7.69)&0.36&0.09&0.89 (5.93)
				\vspace{0.05cm}\\
				1
				&0.42&0.34&0.95 (23.49)&0.38&0.11&0.93 (8.14)&0.30&0.09&0.92 (6.61)
				\vspace{0.05cm}\\
				5
				&0.36&0.31&0.95 (23.41)&0.31&0.07&0.94 (8.83)&0.25&0.06&0.96 (6.60)
				\vspace{0.05cm}\\
				10
				&0.32&0.30&0.95 (23.40)&0.29&0.06&0.95 (9.02)&0.25&0.07&0.95 (6.55)
				\vspace{0.05cm}\\
    			10
				&0.32&0.30&0.95 (23.40)&0.29&0.06&0.95 (9.02)&0.25&0.07&0.95 (6.55)
				\vspace{0.05cm}\\
    			20
				&0.32&0.27&0.95 (23.74)&0.28&0.07&0.94 (9.11)&0.22&0.06&0.95 (6.52)
				\vspace{0.05cm}\\
    			30
                &0.32&0.35&0.95 (24.20)&0.29&0.07&0.94 (9.09)&0.21&0.07&0.95 (6.45)
				\vspace{0.05cm}\\
        		40
                &0.36&0.49&0.96 (24.42)&0.28&0.07&0.94 (8.91)&0.21&0.07&0.94 (6.45)
				\vspace{0.05cm}\\
				50
				&0.33&0.34&0.97 (24.15)&0.29&0.07&0.95 (9.15)&0.25&0.06&0.95 (6.63)
				\vspace{0.05cm}\\
				\hline
			\end{tabular}
		}
\end{subtable}
\begin{subtable}[t]{\linewidth}
\vspace{0.2cm}
\centering
		\resizebox{1.0\textwidth}{!}{
			\begin{tabular}{c|ccc|ccc|ccc|}
				\hline
				&\multicolumn{3}{c|}{Simulation 4}&\multicolumn{3}{c|}{Simulation 5}&\multicolumn{3}{c|}{Simulation 6}\\
				\hline
				$\alpha$&$\mathbb{E}(Y \mid \bm{X})$&$\text{sd}(Y \mid \bm{X})$&Cov (Width)&$\mathbb{E}(Y \mid \bm{X})$&$\text{sd}(Y \mid \bm{X})$&Cov (Width)&$\mathbb{E}(Y \mid \bm{X})$&$\text{sd}(Y \mid \bm{X})$&Cov (Width)\\
				\hline
				\hline
				0.5
                &0.14&0.01&0.96 (3.95)&0.004&0.0001&0.92 (0.38)&6.88&59.65&0.78 (17.08)
				\vspace{0.05cm}\\
				1
				&0.15&0.07&0.95 (4.33)&0.004&0.0001&0.93 (0.42)&7.20&57.75&0.79 (17.72)
				\vspace{0.05cm}\\
				5
				&0.14&0.08&0.96 (4.50)&0.005&0.03&0.99 (1.01)&7.20&56.74&0.79 (18.80)
				\vspace{0.05cm}\\
				10
				&0.13&0.08&0.95 (4.39)&0.007&0.06&0.99 (1.33)&6.63&58.71&0.78 (18.14)
				\vspace{0.05cm}\\
    			20
				&0.13&0.09&0.95 (4.49)&0.007&0.08&0.99 (1.64)&6.33&59.26&0.78 (17.97)
				\vspace{0.05cm}\\
    			30
                &0.14&0.12&0.95 (4.39)&0.008&0.08&0.99 (1.69)&7.07&57.88&0.79 (18.72)
				\vspace{0.05cm}\\
        		40
                &0.13&0.07&0.95
                (4.29)&0.007&0.08&0.99 (1.69)&6.86&57.77&0.79 (18.70)
				\vspace{0.05cm}\\
				50
				&0.15&0.14&0.94 (4.42)&0.009&0.10&0.99 (1.81)&6.69&58.25&0.78 (18.55)
				\vspace{0.05cm}\\
				\hline
			\end{tabular}
		}
\end{subtable}
	\vspace{-0.1cm}
	\caption{\small PGQR results with different choices of $\alpha \in \{ 0.5,1.0,5.0,10.0, 20.0, 30.0, 40.0, 50.0\}$ in the variability penalty for Simulations 1 through 6. This table reports the average PMSE for the conditional expectation and standard deviation, as well as the coverage rate (Cov) and the average width of the 95\% prediction intervals. Results were averaged across 20 replicates.}
\label{tab:robustana}
\end{table}

\end{document}